\pdfoutput=1

\documentclass[sigconf, nonacm]{acmart}
\usepackage{thmtools,thm-restate}
\usepackage{tikz}
\usepackage{pgfplots}
\usetikzlibrary{arrows,intersections}
\usepackage{subcaption}
\usepackage[ruled,lined,linesnumbered]{algorithm2e} 
\usepackage{algorithmic}
\usepackage{mathtools}
\usepackage{enumitem}

\newcommand{\tree}{\tau}

\newcommand{\mec}{\mathcal{M}}
\newcommand{\mecld}{M_{LD}}

\newcommand{\mecexp}{M_{EM}}
\newcommand{\mecpf}{M_{PF}}
\newcommand{\mecrnm}{M_{RNM}}
\newcommand{\domain}{\mathcal{X}}

\newcommand{\range}{\mathcal{R}}
\newcommand{\real}{\mathbb{R}}

\newcommand{\TT}{\mathcal{T}}

\makeatletter
\newcommand{\vast}{\bBigg@{3}}
\newcommand{\vvast}{\bBigg@{4}}
\newcommand{\Vast}{\bBigg@{5}}

\makeatother

\newenvironment{customlegend}[1][]{
\begingroup
\csname pgfplots@init@cleared@structures\endcsname
\pgfplotsset{#1}
}{%
\csname pgfplots@createlegend\endcsname
\endgroup
}%
\def\addlegendimage{\csname pgfplots@addlegendimage\endcsname}

\newcommand\vldbdoi{XX.XX/XXX.XX}
\newcommand\vldbpages{XXX-XXX}
\newcommand\vldbvolume{14}
\newcommand\vldbissue{1}
\newcommand\vldbyear{2020}
\newcommand\vldbauthors{\authors}
\newcommand\vldbtitle{\shorttitle} 
\newcommand\vldbavailabilityurl{https://github.com/victorfarias/MultiObjectiveDifferentiallyPrivateSelection}
\newcommand\vldbpagestyle{plain} 

\begin{document}
\title{Differentially Private Multi-Objective Selection: Pareto and Aggregation Approaches}

\author{Victor A. E. Farias}
\affiliation{
  \institution{Universidade Federal do Ceará}
  \city{Fortaleza}
  \state{Ceará}
  \country{Brazil}
}
\email{victor.farias@lsbd.ufc.br}

\author{Felipe T. Brito}
\affiliation{
  \institution{Universidade Federal do Ceará}
  \city{Fortaleza}
  \state{Ceará}
  \country{Brazil}
} 
\email{felipe.timbo@lsbd.ufc.br}

\author{Cheryl Flynn}
\affiliation{%
  \institution{AT\&T Chief Data Office}
  \city{Bedminster}
  \state{NJ}
  \country{USA}
}
\email{cflynn@research.att.com}

\author{Javam C. Machado}
\affiliation{
  \institution{Universidade Federal do Ceará}
  \city{Fortaleza}
  \state{Ceará}
  \country{Brazil}
} 
\email{javam.machado@lsbd.ufc.br}

\author{Divesh Srivastava}
\affiliation{
  \institution{AT\&T Chief Data Office}
  \city{Bedminster}
  \state{NJ}
  \country{USA}
}
\email{divesh@research.att.com}

\begin{abstract}
  Differentially private selection mechanisms are fundamental building blocks for privacy-preserving data analysis. While numerous mechanisms exist for single-objective selection, many real-world applications require optimizing multiple competing objectives simultaneously. We present two novel mechanisms for differentially private multi-objective selection: PrivPareto and PrivAgg. PrivPareto uses a novel Pareto score to identify solutions near the Pareto frontier, while PrivAgg enables privacy-preserving weighted aggregation of multiple objectives. Both mechanisms support global and local sensitivity approaches, with comprehensive theoretical analysis showing how to compose sensitivities of multiple utility functions. We demonstrate their practical applicability through two real-world applications: cost-sensitive decision tree construction and multi-objective influential node selection in social networks. The experimental results showed that our local sensitivity-based approaches achieve significantly better utility compared to global sensitivity approaches across both applications and both Pareto and Aggregation approaches. Moreover, the local sensitivity-based approaches are able to perform well with typical privacy budget values $\epsilon \in [0.01, 1]$ in most experiments.

\end{abstract}

\maketitle

\pagestyle{\vldbpagestyle}
\begingroup\small\noindent\raggedright\textbf{PVLDB Reference Format:}\\
\vldbauthors. \vldbtitle. PVLDB, \vldbvolume(\vldbissue): \vldbpages, \vldbyear.\\
\href{https://doi.org/\vldbdoi}{doi:\vldbdoi}
\endgroup
\begingroup
\renewcommand\thefootnote{}\footnote{\noindent
This work is licensed under the Creative Commons BY-NC-ND 4.0 International License. Visit \url{https://creativecommons.org/licenses/by-nc-nd/4.0/} to view a copy of this license. For any use beyond those covered by this license, obtain permission by emailing \href{mailto:info@vldb.org}{info@vldb.org}. Copyright is held by the owner/author(s). Publication rights licensed to the VLDB Endowment. \\
\raggedright Proceedings of the VLDB Endowment, Vol. \vldbvolume, No. \vldbissue\ %
ISSN 2150-8097. \\
\href{https://doi.org/\vldbdoi}{doi:\vldbdoi} \\
}\addtocounter{footnote}{-1}\endgroup

\ifdefempty{\vldbavailabilityurl}{}{
\vspace{.3cm}
\begingroup\small\noindent\raggedright\textbf{PVLDB Artifact Availability:}\\
The source code, data, and/or other artifacts have been made available at \url{\vldbavailabilityurl}.
\endgroup
}

\vspace{-10pt}

\providetoggle{vldb}
\settoggle{vldb}{false}

\section{Introduction}
\label{sec:introduction}


Data analysis is a fundamental tool for decision making in many areas, such as healthcare, finance, and marketing. Data analysts can use data to extract valuable insights, make predictions, optimize decisions and publish their findings. On the other hand, a malicious adversary can use the published data to infer sensitive information about individuals \cite{henriksen2016re, hamza2013attacks, narayanan2016break}. 

To this end, Differential Privacy (DP) \cite{dwork2011differential, dwork2006calibrating} arises as a formal definition that provides strong privacy guarantees for data release. It supposes that a malicious adversary has some prior knowledge about $n-1$ individuals and wants to infer the information about the $n$-th individual. Differential privacy ensures that the adversary's belief about the n-th individual does not change significantly based on the presence or absence of the individual in the dataset. 

Algorithms can achieve differential privacy by perturbing the output, which releases an approximate answer to a non-private query over an input database with noise added. The magnitude of the noise should be large enough to cover the identity of the individuals in the input database.

Differential privacy has been widely applied to the problem of \textit{single-objective private selection}. In this problem, the goal is to select an element from a set of candidates that maximizes a single objective function. Formally, given a set of candidates $\range$ and a utility function $u: \domain \times \range \rightarrow \real$ that takes a candidate $r \in \range$ and a dataset $x \in \domain$ and outputs a score, the goal is to select an element $r \in \range$ that approximately maximizes $u(x, r)$ while satisfying differential privacy. The utility function $u$ is designed by the data analyst and measures how well a candidate $r$ satisfies the objective of a given task. For the rest of the paper, we will use utility function and objective function interchangeably.

Many algorithms have been proposed to solve this problem under differential privacy: the exponential mechanism \cite{mcsherry2007mechanism}, the report noisy max algorithm \cite{dwork2014algorithmic}, the permute-and-flip mechanism \cite{mckenna2020permute} and the local dampening mechanism \cite{DBLP:journals/pvldb/FariasBFMMS20, farias2023local}. These mechanisms use the notion of sensitivity in two flavors, \textit{global sensitivity} and \textit{local sensitivity}, where these notions measure how much the utility function changes when a single individual is added or removed from a dataset.

Such algorithms are very important in the field of differential privacy since they are fundamental building blocks and can be composed to create more complex algorithms. Their applications include private release of streams, histograms, databases and graphs \cite{zhu2017differentially}, data mining \cite{fletcher2019decision}, machine learning \cite{ji2014differential, liu2021machine} and social network analysis \cite{task2012guide, jiang2021applications}.

However, many real world problems require the optimization of multiple competing objectives. For instance, in healthcare, a medical diagnostic machine learning system for detecting life-threatening diseases must balance multiple competing objectives. A classifier must achieve high true positive rate (TPR) to correctly identify patients who have the disease, while also maintaining high true negative rate (TNR) to avoid unnecessary treatments and anxiety in healthy patients. A false negative could mean missing a critical diagnosis, while a false positive could lead to unnecessary invasive procedures. Moreover, when releasing such diagnostic models, protecting patient privacy through differential privacy is essential to prevent the leakage of sensitive health information. 

Moreover, releasing a machine learning model trained on patient data could inadvertently leak sensitive information about individuals in the training dataset. Recent attacks have shown that it's possible to extract training data from machine learning models through membership inference attacks \cite{shokri2017membership}. Therefore, differential privacy plays a crucial role in preventing such attacks by adding carefully calibrated noise during the model training and selection process, ensuring that the presence or absence of any individual patient's data cannot be reliably detected in the final model.



Thus, we define our problem statement of \textit{multi-objective private selection} as follows. Given a set of candidates $\range$ and a set of utility functions $u_1, u_2, \ldots, u_m: \domain \times \range \rightarrow \real$, the goal is to select an element $r \in \range$ that maximizes $\{u_1(x,r), u_2(x,r), \ldots, u_m(x,r)\}$ while satisfying differential privacy.

Note that the utility functions may be conflicting and there could be no candidate that maximizes all utility functions simultaneously. Instead, it is desired to find a set of candidates that are not dominated in all utility functions. These are said to be \textit{Pareto optimal} and the set of all Pareto optimal candidates is called the \textit{Pareto front}. In the example of decision trees, the Pareto front would contains trees such that no other tree is better in both true positive rate and true negative rate.

A second approach for tackling multi-objective problems is to aggregate the objective functions into a single scalar objective function using the weighted sum method. In this scenario, the data analyst assigns a weight to each objective function and the goal is to select a candidate that maximizes the weighted sum of the objective functions. This approach is useful in the case that the data analyst has prior knowledge about the importance of each objective function. Moreover, it can be seen as complementary to the Pareto approach where the data analyst can aggregate a subset of the objective functions that make sense to be aggregated, e.g., objective functions that are related to cost in the same currency, and use the Pareto approach for the aggregated objective function with the remaining objective functions.

In this paper, we propose two novel approaches to solve the problem of multi-objective private selection under differential privacy. We first propose the \textit{PrivPareto} mechanism, which selects candidates that are as close as possible to the Pareto front under differential privacy. This mechanism gives preference for the candidates that are the least dominated by other candidates, i.e., that there are the least number of candidates that are better in all objective functions. For that, we propose the \textit{Pareto score} that measures how many candidates are better than a given candidate in all objective functions. We then propose the \textit{PrivAgg} mechanism where it aims to approximately select the best candidate with the largest weighted sum of the utility functions under differential privacy. 

A single objective mechanism requires the data analyst to provide the utility function and the global or local sensitivity. Similarly, in our setting, the data analyst also provides the multiple utility functions and the global or local sensitivity for each utility function. Our mechanisms support both global and local sensitivity. Single objective mechanisms are part of the definition of our mechanisms for selecting the best candidate based on the Pareto score or on the weighted sum of the utility functions. For that, our mechanisms algorithmically compute the resulting sensitivity of the Pareto score and weighted sum which frees the data analyst from the burden of computing the sensitivity manually.

To show the effectiveness and generality of our mechanisms, we conduct an extensive experimental evaluation on two distinct problems: (i) multi-objective influential node selection in social networks where we select the most influencial nodes in a graph based on multiple centrality measures, and (ii) cost-sensitive decision tree construction where we build a set of approximate Pareto optimal decision trees maximizing true positive rate and true negative rate under differential privacy.


To the best of our knowledge, this is the first work to address the problem of multi-objective private selection using differential privacy. Our contributions are summarized as follows:

\begin{itemize}
    \item We propose PrivPareto, a novel multi-objective selection mechanism. We propose a Pareto based score that indicates how much a given candidate is dominated, i.e., how many candidates are better than the given candidate in all objective functions. This score represents how good is a candidate when compared to the the other candidates.
    \item We develop the algorithm to compute the sensitivity of the Pareto score based on the utility functions and the global or local sensitivity of the utility functions provided by the data analyst. 
    \item  We propose PrivAgg, a multi-objective selection mechanism based on weighted sum aggregation of the utility functions into a higher scalar utility function. 
    \item  We develop a theoretical background to show how to algorithmically compose the global and local sensitivity of the higher aggregated scalar function based on the single utility functions.    
    \item We address the problem of constructing cost sensitive decision trees with two conflicting utility functions: true positive rate and true negative rate. We tackle this problem using differentially private evolutionary family of algorithms called \textit{DP-MOET} with the PrivPareto and PrivAgg mechanisms. For that, we also show how to compute global and local sensitivity for true positive rate and true negative rate. 
    \item We create the family of algorithms called \textit{DP-MoTkIN} applying PrivPareto and PrivAgg mechanism to tackle  problem of influential analysis on graphs where the goal is to retrieve the top-k most influential nodes. As influence metric, we use two conflicting metrics: egocentric density and degree centrality. Also, we address the computation of the local and global sensitivity for egocentric density and degree.
\end{itemize}


The paper is organized as follows: Section 2 covers the background on differential privacy and sensitivity analysis. Section 3 reviews related work on differentially private selection mechanisms. Sections 4 and 5 introduce our main contributions, PrivPareto and PrivAgg, with theoretical analyses. Section 6 applies these mechanisms to cost-sensitive decision tree construction and Section 7 to multi-objective influential node selection in social networks. Section 8 concludes with a summary of contributions and results.

\section{Differential Privacy}
\label{sec:background}

In this section, we present the basic definitions of differential privacy, highlighting sensitivity metrics (global and local) and defining sensitivity functions that provide efficient upper bounds for computing local sensitivity.

\subsection{Basic Definitions}

Let $x$ represent a sensitive database and \( f \) a function (or query) to be executed on \( x \). The database is modeled as a vector \( x \in \mathcal{D}^n \), with each entry corresponding to an individual tuple. To release the output \( f(x) \) without compromising the privacy of individuals, it is necessary to design a randomized algorithm \( \mathcal{M}(x) \) that introduces noise to \( f(x) \) while adhering to the formal definition of differential privacy outlined below.

\begin{definition}
    ($\epsilon$-Differential Privacy \cite{dwork2006our, dwork2006calibrating}). A randomized algorithm $\mathcal{M}$ satisfies $\epsilon$-differential privacy if, for any two databases $x$ and $y$ such that $d(x, y) \leq 1$, and for any possible output $O$ of $\mathcal{M}$, the following inequality holds:
    \[
    Pr[\mathcal{M}(x) = O] \leq \exp(\epsilon) Pr[\mathcal{M}(y) = O],
    \]
    where $Pr[\cdot]$ denotes the probability of an event, and $d$ is the Hamming distance between the two databases, defined as $d(x, y) = |\{ i \mid x_i \neq y_i \}|$, i.e., the number of tuples that differ. We refer to $d$ as the distance between the two databases.
    \label{def:dp}
\end{definition}


The parameter $\epsilon$, called the \textit{privacy budget}, controls the similarity of the output distributions between databases $x$ and $y$. Smaller $\epsilon$ values ensure closer distributions, enhancing privacy but reducing accuracy, while larger $\epsilon$ values improve accuracy at the cost of weaker privacy. 
Multiple queries to the database can progressively consume the privacy budget. In this context, sequential composition comes into play when multiple mechanisms are applied to a dataset, with the total privacy budget being the sum of the budgets allocated to each computation:

\begin{theorem}
    (Sequential composition \cite{mcsherry2007mechanism}) Let $\mec_i : \domain \rightarrow \range_i$ be an $\epsilon_i$-differentially private algorithm for $i \in [k]$. Then $\mec(x)=(\mec_1(x),\cdots,\mec_k(x))$ is ($\sum_{i=1}^{k}$)-differentially private.
    \label{theorem:sequential_composition}
\end{theorem}


\subsection{Sensitivity}
\label{sec:sensitivity}

Differentially private mechanisms usually perturb the true output with noise. The amount of noise added to the true output of a non-numeric function $f: \domain \rightarrow \range$ is proportional to the \textit{sensitivity} of the utility function $u: \mathcal{D}^n \times \mathcal{R} \rightarrow \mathbb{R}$. 


\textbf{Global Sensitivity}. The global sensitivity of $u$ quantifies the maximum difference in utility scores across all possible pairs of neighboring database inputs $x, y$ and all elements $r \in \mathcal{R}$:

\begin{definition}
    (Global Sensitivity \cite{mcsherry2007mechanism}). Given a utility function $u: \mathcal{D}^n \times \mathcal{R} \rightarrow \mathbb{R}$ that takes as input a database $x \in \domain$ and an element $r \in \range$ and outputs a numeric score for $r$ in $x$. The global sensitivity of $u$ is defined as:
    $$ \Delta u = \max_{r \in \mathcal{R}} \max_{x,y | d(x,y) \leq 1} |u(x,r) - u(y,r)|.$$ 
    \label{def:global_sensitivity}
\end{definition}

Global sensitivity represents the maximum possible variation in the utility score for any $r \in \mathcal{R}$ across all neighboring database pairs, providing a worst-case measure of sensitivity for DP mechanisms.

\textbf{Local Sensitivity}. In private selection mechanisms, higher global sensitivity $\Delta u$ often leads to reduced accuracy. To address this, private solutions aim for mechanisms with lower sensitivity. The concept of local sensitivity $LS(x)$ \cite{nissim2007smooth} measures sensitivity specifically at the given input database $x$, rather than considering the entire universe of databases $\mathcal{D}^n$.

\begin{definition}
    (Local Sensitivity \cite{nissim2007smooth, farias2023local}). Given a utility function $u(x,r)$ that takes as input a database $x$ and an element $r$ and outputs a numeric score, the local sensitivity of $u$ is defined as $$ LS^{u}(x) = \max_{r \in \mathcal{R}} \max_{y| d(x,y) \leq 1}|u(x,r)-u(y,r)|$$
\end{definition}

Local sensitivity is often smaller than global sensitivity for many practical problems \cite{blocki2013differentially, karwa2011private, kasiviswanathan2013analyzing, lu2014exponential,nissim2007smooth, zhang2015private}. In such cases, real-world databases typically differ significantly from the worst-case scenarios considered in global sensitivity, resulting in a much lower observed local sensitivity. Note that global sensitivity represents the maximum local sensitivity in all databases, expressed as $\Delta u = \max_{x} LS^u(x)$.

However, relying exclusively on local sensitivity is insufficient to guarantee differential privacy. Therefore, a complementary concept, called local sensitivity at distance $t$, is introduced:

\begin{definition}\label{def:locsen_dist_t}
	(Local Sensitivity at distance $t$ \cite{nissim2007smooth, farias2023local}). Given a utility function $u: \mathcal{D}^n \times \mathcal{R} \rightarrow \mathbb{R}$ that takes as input a database $x \in \domain$ and an element $r \in \range$ and outputs a numeric score for $r$ in $x$, the local sensitivity at distance $t$ of $u$ is defined as
	$$LS^{u}(x,t) = \max_{y| d(x,y) \leq t} LS^{u}(y).$$
\end{definition}

Local sensitivity at distance $t$, $LS^u(x,t)$, captures the maximum local sensitivity $LS^u(y)$ in all databases $y$ within a distance of $t$. In other words, it considers up to $t$ modifications to the database before evaluating its local sensitivity. In particular, $LS^u(x,0) = LS^u(x)$.

A limitation of the definition of local sensitivity at distance $t$ arises when a single element in $\mathcal{R}$ exhibits a high sensitivity value (close to $\Delta u$), causing $LS^u(x, t)$ to also be large. This becomes problematic when some outliers inflate $LS^u(x, t)$ despite the fact that most elements have low sensitivity, reducing accuracy. To address this, individual sensitivity measurements can be performed for each element in $\mathcal{R}$ using the concept of element local sensitivity:

\begin{definition}
	(Element Local Sensitivity at distance $t$ \cite{farias2023local}). Given a utility function $u(x,r)$ that takes as input a database $x$ and an element $r$ and outputs a numeric score for $x$, the element local sensitivity at distance $t$ of $u$ is defined as 
	$$ LS^{u}(x,t,r) = \max_{y \in \domain| d(x,y) \leq t, z \in \domain | d(y,z) \leq 1}|u(y,r)-u(z,r)|,$$	
	where $d(x,y)$ denotes the distance between two databases.
	\label{def:item_local}
\end{definition} 

Note that we can obtain $LS^u(x, t)$ from this definition: $ LS^{u}(x,t) = \max_{r \in \mathcal{R}} LS^{u}(y,t,r)$ as $LS^{u}(x,t,r) = \max_{y| d(x,y) \leq t} LS^{u}(y,0,r)$.

\subsection{Sensitivity Functions}
\label{sec:sensitivity_functions}

Computing local sensitivity $LS^u(x,t)$ or element local sensitivity $LS^u(x,t,r)$ can be infeasible is some cases where it is an NP-hard problem \cite{nissim2007smooth, zhang2015private}. To overcome this, we can build a computationally efficient function $\delta^u(x,t,r)$ that computes as an upper bound for $LS^u(x,t)$ or $LS^u(x,t,r)$, while staying below $\Delta u$. These functions are denoted as sensitivity functions. A sensitivity function needs to have the following signature: $\delta^u: \domain \times \mathbb{N} \times \range \rightarrow \real$ and the function $\delta^u$ is said to be a sensitivity function for the utility function $u$. The $\delta^u(x,t,r)=\Delta u$, $\delta^u(x,t,r)=LS^{u}(x,t)$ and $\delta^u(x,t,r)=LS^{u}(x,t,r)$ are all sensitivity functions.

Additionally, a sensitivity function must satisfy the property of \textit{admissibility} to ensure compliance with DP in single-objective selection mechanisms. This property requires the sensitivity function to serve as an upper bound on the element local sensitivity.

\begin{definition}
	(Admissibility \cite{farias2023local}). A sensitivity function $\delta^{u}(x, t,$ $ r)$ is \textit{admissible} if:
	\begin{enumerate}
		\item $\delta^{u}(x, 0, r) \geq LS^{u}(x, 0, r)$, for all $x \in \domain$ and all $r \in \range$
		\item $\delta^{u}(x, t+1, r) \geq \delta^{u}(y, t, r)$, for all $x,y$ such that $d(x,y) \leq 1$ and all $t \geq 0$		
	\end{enumerate}
	\label{def:admissible_function}
\end{definition}

A key result is that the global sensitivity $\delta^u(x,t,r)=\Delta u$, local sensitivity at distance $t$ $\delta^u(x,t,r)=LS^u(x,t)$, and element local sensitivity at distance $t$ $\delta^u(x,t,r)=LS^u(x,t,r)$ are all admissible sensitivity functions \cite{farias2023local}. 


For the remainder of this paper, we utilize the concept of admissible sensitivity functions $\delta^u(x, t, r)$ to construct differentially private mechanisms, as they are generic functions capable of representing various sensitivity definitions or serving as upper bounds for these sensitivities.

\section{Related Work}
\label{sec:related_work}

The field of differentially private multi-objective selection remains unexplored. To the best of our knowledge, this work is the first to tackle this problem. Nonetheless, we review the literature on related works in differential privacy for single-objective selection.



\subsection{Exponential Mechanism}

The exponential mechanism $\mecexp$ \cite{mcsherry2007mechanism} is one of the most widely used approaches for ensuring differential privacy in single-objective selection problems. It leverages the concept of global sensitivity $\Delta u$ (see Definition \ref{def:global_sensitivity}) to manage privacy guarantees.

This mechanism selects a candidate from the set $\mathcal{R}$ by sampling an element $r \in \mathcal{R}$ with probability proportional to its utility score $u(x, r)$ computed on the database $x$. By employing an exponential distribution, it assigns probabilities to each $r \in \mathcal{R}$ based on their utility scores, ensuring DP. The formal definition is as follows:

\begin{definition}
	(Exponential Mechanism \cite{mcsherry2007mechanism}). The exponential mechanism $\mecexp(x, \epsilon, u, \Delta u, \mathcal{R})$ selects and outputs an element $r \in \mathcal{R}$ with probability proportional to $\exp \big(\frac{\epsilon \; u(x, r)}{2 \Delta u}\big)$.
	\label{def:exponential_mechanism}
\end{definition}


The exponential mechanism satisfies $\epsilon$-differential privacy \cite{mcsherry2007mechanism}. 




\subsection{Permute-and-Flip}

\
The \textit{permute-and-flip} mechanism $\mecpf$ \cite{mckenna2020permute} is an alternative to the exponential mechanism for the task of differentially private selection. It preserves the properties of the exponential mechanism while offering improved performance, with an expected error that is never higher and can be reduced by up to half. This makes it a robust alternative for single-objective selection tasks under DP.


The mechanism operates by randomly permuting the elements of $\mathcal{R}$ and flipping a biased coin for each element, with the bias determined by the utility score $u(x, r)$ relative to the maximum utility $u^*$. This design prioritizes elements with higher utility scores, increasing the likelihood of selecting optimal candidates. Importantly, the mechanism guarantees termination since elements with utility equal to $u^*$ always have a coin flip probability of 1.


\subsection{Report-noisy-max}

The \textit{report-noisy-max} mechanism $\mecrnm$ \cite{dwork2014algorithmic} is a versatile approach for private selection problems. It adds independent numeric noise to the utility score of each candidate and selects the candidate with the highest noisy score. Formally, the mechanism is defined as:

\[
\mecrnm(x, \epsilon, u, \Delta u, \mathcal{R}) = \arg\max_{r \in \mathcal{R}} \big( u(x, r) + Z \big),
\]

\noindent where $Z$ represents noise drawn from a specified distribution.

The mechanism supports various noise distributions, including Laplace, Gumbel, and Exponential. Initially, it was proposed using the Laplace distribution, with $Z \sim \texttt{Lap}\left( \frac{\epsilon}{2\Delta u} \right)$, ensuring a balance between privacy and utility. 

Additionally, $\mecrnm$ can utilize the Exponential distribution, where $Z \sim \texttt{Expo}\left( \frac{\epsilon}{2\Delta u} \right)$ \cite{ding2021permute}. Under this configuration, it becomes equivalent to the \textit{permute-and-flip} mechanism \cite{mckenna2020permute}. When employing the Gumbel distribution, $\texttt{Gumbel}\left( \frac{2\Delta u}{\epsilon} \right)$, the \textit{report-noisy-max} mechanism mirrors the exponential mechanism \cite{durfee2019practical}. 

\subsection{Local Dampening Mechanism}

The local dampening mechanism $\mecld$ \cite{farias2023local} is a single-objective selection approach that leverages local sensitivity to minimize the noise introduced during selection. By dampening the utility scores of candidates, it improves the signal-to-noise ratio, resulting in reduced noise. The mechanism is formally defined as follows:

\begin{definition}
	(Local dampening mechanism). The local dampening mechanism $\mecld(x, \epsilon, u, \delta^{u}, \mathcal{R})$ selects and outputs an element $r \in \mathcal{R}$ with probability proportional to $\exp \big(  \frac{\epsilon \; D_{u,\delta^u}(x, r)}{2}\big) $.
\end{definition}

\noindent Here $D_{u,\delta^u}(x,r)$ is the dampening function defined as:

\begin{definition}
	(Dampening function). Given a utility function $u(x, r)$ and an admissible function $\delta^u(x, t, r)$, the dampening function $D_{u,\delta^u}(x,r)$ is defined as a piecewise linear interpolation over the points:	    
	$$< \ldots,(b(x,-1,r), -1),(b(x,0,r), 0),(b(x,1,r), 1), \ldots >$$
	where $b(x, i, r)$ is given by:		    	    
	\begin{equation*}	
		b(x, i, r) \coloneqq
		\begin{cases}
			\sum^{i-1}_{j=0} \delta^u(x, j, r) & \text{if} \ i > 0 \\
			0 & \text{if} \ i = 0 \\
			- b(x, -i, r)                    & \text{otherwise}     
		\end{cases}
	\end{equation*}
	Therefore,	    
	\begin{equation*}
		D_{u,\delta^u}(x, r) =  \frac{u(x, r)-b(x,i,r)}{b(x, i+1, r)- b(x, i, r)} + i
	\end{equation*}					
	with $i$ being defined as the smallest integer such that $u(x, r) \in  \left[ b(x,i,r), b(x,i+1,r)  \right)$.
\end{definition}

\section{PrivPareto Mechanism}
\label{sec:privpareto}

We introduce the PrivPareto mechanism, a multi-objective differentially private algorithm for selecting an approximate Pareto optimal solution from a set of candidate solutions.


We propose a high level score based on the Pareto dominance relation between the candidates called \textit{Pareto Score} $PS(x,r)$. This score uses the notion of domination, i.e., a candidate $r$ dominates another candidate $r'$ if $r$ is better than $r'$ in all utility functions in a given dataset $x$. The Pareto Score indicates how dominated a given element is by the other elements in the set. We discuss it in detail in Section \ref{sec:pareto_score}.

The PrivPareto mechanism is composed by three main phases: (i) Pareto Score computation; (ii) Pareto score sensitivity computation and (iii) private selection of the approximate Pareto optimal solution.

The first step is the computation of the Pareto score $PS(x,r)$ for each candidate $r \in \range$ given the input database $x$. In the second step, we compute the sensitivity of the Pareto score to use in the next step, the selection phase. In the third step, we select the approximate Pareto optimal solution based on the Pareto Score using a single objective differentially private selection mechanism. 

The choice of the single objective differentially private selection mechanism in the third step plays an important role on the accuracy of the our mechanism since each mechanism adds a different amount of noise to the output. These mechanisms can use global or local sensitivity. Global sensitivity based algorithms are more frequent in the literature \cite{mcsherry2007mechanism, mckenna2020permute}. However, local sensitivity based mechanisms have been shown to perform better in applications where the local sensitivity is smaller than the global sensitivity \cite{farias2023local} which is the case for the Pareto Score. Thus, we present the PrivPareto mechanism in two versions: a global sensitivity version and a local sensitivity version.

The global sensitivity version of the PrivPareto mechanism is defined in the Definition \ref{def:global_privpareto}.

\begin{definition}

    (Global PrivPareto Mechanism) The global PrivPareto mechanism $PrivPareto_{global} (x, \epsilon, u_1, \ldots, u_m, \mathcal{R})$ takes as input a database $x \in \domain$, the privacy budget $\epsilon$, the utility functions $u_1, u_2, \ldots, u_m: \domain \times \range \rightarrow \real$, the set of candidate solutions $\range$, and outputs an approximate Pareto optimal solution $r \in \range$. The mechanism is defined as follows:    
    $$ PrivPareto_{global}(x, \epsilon, u_1, \ldots, u_m, \mathcal{R}) = \mec_{global}(x, \epsilon, PS,\Delta PS, \mathcal{R}). $$

\label{def:global_privpareto}
\end{definition}

\noindent Here $\mec_{global}$ is a single objective differentially private selection mechanism that uses the global sensitivity as the exponential mechanism, permute-and-flip mechanism and the report-noisy-max mechanism (see Section \ref{sec:related_work}), $\Delta PS$ is the global sensitivity of the Pareto Score, and $PS$ is the Pareto Score. We discuss the global sensitivity of the Pareto Score in Section \ref{sec:sensitivity_pareto_score}.

The local sensitivity version of the PrivPareto mechanism is given in Definition \ref{def:local_privpareto}:

\begin{definition}

    (Local PrivPareto Mechanism) The local PrivPareto mechanism $PrivPareto_{local} (x, \epsilon, u_1, \ldots, u_m, \delta^{u_1}, \ldots, \delta^{u_m}, \mathcal{R})$ takes as input a database $x \in \domain$, the privacy budget $\epsilon$, the set of utility functions $u_1, u_2, \ldots, u_m: \domain \times \range \rightarrow \real$, the admissible sensitivity functions $\delta^{u_1} , \ldots, \delta^{u_m}$ of the utility functions, a set of candidate solutions $\range$, and outputs an approximate Pareto optimal solution $r \in \range$. The mechanism is defined as follows:
    \begin{multline*}
        PrivPareto_{local} (x, \epsilon, u_1, \ldots, u_m, \delta^{u_1}, \ldots, \delta^{u_m}, \mathcal{R})  \\ = \mec_{local}(x, \epsilon, PS, \delta^{PS}, \mathcal{R}).  
    \end{multline*}

\label{def:local_privpareto}
\end{definition}

\noindent Here $\mec_{local}$ is a single objective differentially private selection mechanism that uses the local sensitivity, $\delta^{PS}$ is the sensitivity function of the Pareto Score (discussed in Section \ref{sec:sensitivity_pareto_score}), and $PS$ is the Pareto Score. Note that the admissible sensitivity functions $\delta^{u_1}, \ldots, \delta^{u_m}$ and $\delta^{PS}$ can represent the global sensitivity, the local sensitivity at distance t or the element local sensitivity at distance t (see Sections \ref{sec:sensitivity} and \ref{sec:sensitivity_functions}).



\subsection{Pareto Score}
\label{sec:pareto_score}

The Pareto Score is a numeric score that represents how close a candidate solution is to the Pareto frontier. The Pareto Score is computed based on the utility functions $u_1, u_2, \ldots, u_m$ and the database $x$. 

The Pareto Score is a high level score that uses the notion of domination, i.e., a candidate $r$ dominates another candidate $r'$ if $r$ is larger than $r'$ in all utility functions. The Pareto Score indicates how dominated a given element is by the other elements in the set. We say that $a$ is Pareto optimal if there is no other candidate that dominates $a$. The set of all Pareto optimal candidates is called the Pareto front.

We say the a candidate $a$ \textit{dominates} a candidate $b$ on a database $x$, $a \succeq^x b$, if and only if $u_i(x, a) \geq u_i(x, b)$ for all $i \in [m]$. It is said that $a \nsucceq^x b$ if $a$ does not dominate $b$. Thus, we define the Pareto Score as follows:

\begin{definition}
    (Pareto Score). Given a database $x \in \domain$, a set of utility functions $u_1, u_2, \ldots, u_m: \domain \times \range \rightarrow \real$, and a candidate solution $r \in \range$, the Pareto Score $PS(x,r)$ is defined as the number of candidates that dominate $r$ in the set $\range$:
    $$ PS(x,r) = -|\{ r' \in \range \; | \; r' \succeq^x r \}|$$
    \label{def:pareto_score}
\end{definition}

Note that the Pareto score is a non-positive number. It ranges from $0$, for a Pareto optimal candidate, to $-(|\range|-1)$, for an candidate dominated by all the other candidates. 


\begin{example}
    Consider a database $x$, a set of candidates $\range = \{a, b, c, d, e\}$ and two utility functions $u_1, u_2: \domain \times \range \rightarrow \real$. The utility functions are defined as follows:
    \begin{align*}
        u_1(x, a) = 3, u_1(x, b) = 5, u_1(x, c) = 4, u_1(x, d) = 2, u_1(x, e) = 1 \\
        u_2(x, a) = 5, u_2(x, b) = 3, u_2(x, c) = 2, u_2(x, d) = 4, u_2(x, e) = 1
    \end{align*}
    As depicted in Figure \ref{fig:pareto_example}, the score for Pareto optimal $a$ and $b$ is $0$ since they are dominated by no other candidate. The score for $c$ is $-1$ since it is dominated only by $b$. The score for $d$ is $-1$ since it is dominated by $a$. The score for $e$ is $-4$ since it is dominated by all the other candidates.
    
    \begin{figure}
    \centering

  \begin{tikzpicture}[
    thick,
    >=stealth',
    dot/.style = {
      draw,
      fill = white,
      circle,
      inner sep = 0pt,
      minimum size = 14pt
    }
  ]
  \coordinate (O) at (0,0);

  \draw[->] (-0.3,0) -- (6,0) coordinate[label = {below:$u_1(x,r)$}];
  \draw[->] (0,-0.3) -- (0,6) coordinate[label = {right:$u_2(x,r)$}];

  \foreach \x [count=\xi from 1] in {1, 2, 3, 4, 5}{
    \draw [-, thick] (\x,0.1) -- (\x,-0.1);
    \draw [-, thick] (0.1, \x) -- (-0.1, \x);
    \node at (\x, -0.3) { $\x$};
    \node at (-0.3, \x) { $\x$};
  }

  \draw [black, fill=gray!20] plot [smooth cycle] coordinates {(5.0,2.5) (5.5,3.0) (3,5.5) (2.5,5) };
  
  \node [rotate=-45] at (4.5, 4.5) {Pareto Front};
  
  \foreach \x/\y/\score/\label [count=\xi from 1] in {3/5/0/a, 5/3/0/b, 4/2/-1/c, 2/4/-1/d, 1/1/-4/e}{
    \draw [dashed] (\x,\y) -- (0,\y);
    \draw [dashed] (\x,\y) -- (\x,0);
    \draw (\x,\y) node[dot] {\score};
    \node at (\x-.35, \y+0.35) {$\label$};
  }
  
\end{tikzpicture}

\caption{Example Pareto scores for 5 elements and 2 utility functions. The left lower subspace of each candidate, delimited by the two dashed lines, contains its dominated elements. For instance, candidate $a$ dominates $d$ and $e$.}
\label{fig:pareto_example}
\end{figure}
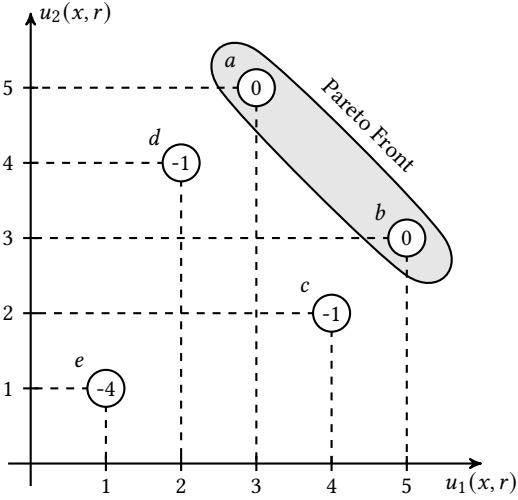

\end{example}

\subsection{Pareto Score Sensitivity}
\label{sec:sensitivity_pareto_score}

Here, we discuss the global and local sensitivity for Pareto Score that are used in the PrivPareto mechanism.

\paragraph{Global Sensitivity} The global sensitivity is $|\range|-1$ for an arbitrary utility function. The proof of this claim is straightforward given that the range of the Pareto score is $[-(|\range|-1), 0]$ and the score of a candidate in a dataset can change from $0$ to $|\range|-1$ in a neighboring dataset. 

\paragraph{Local Sensitivity} For specific utility functions, the global and local sensitivities can be smaller. For this, we show how to compute the local sensitivity of the Pareto score using the local sensitivity of the underlying sensitivity functions. More specifically, we present the admissible sensitivity function $\delta^{PS}(x,t,r)$ that computes an upper bound on the element local sensitivity of the Pareto Score based on the underlying admissible sensitivity functions $\delta^{u_1}(x,t,r), \ldots, \delta^{u_m}(x,t,r)$ of the underlying utility functions.

\begin{definition} [Admissible Function for Pareto Score PS]
    $$\delta^{PS}(x,t,r) = |dom^-_{x,t,r}| + |ndom^+_{x,t,r}|$$
    where:
    \begin{align*}
        dom^-_{x,t,r} & = \{ r' \in dom_{x,r} \; | \; \exists i \in [m] \; s.t. \; u_i^{-t}(x,r') \leq  u_i^{+t}(x,r) \} \\
        ndom^+_{x,t,r} & = \{ r' \in ndom_{x,r} \; | \; u_i^{+t}(x,r')  \geq  u_i^{-t}(x,r), \forall i \in [m] \} \\
        dom_{x,r} & = \{ r' \in \range \; | \; r' \succeq^x r\} \\
        ndom_{x,r} & = \{ r' \in \range \; | \; r' \nsucceq^x r\}
    \end{align*}
    where $u_i^{+t}(x,r)$ be a shorthand for $u_i(x,r) + \sum_{i=0}^{t}\delta^{u_i}(x,i,r)$, $u_i^{-t}(x,r)$ be a shorthand for $u_i(x,r) - \sum_{i=0}^{t}\delta^{u_i}(x,i,r)$ and $[m]$ denote the set $\{1,2,\cdots,m\}$. Note that, using this shorthand, $PS(x,r)=-|dom_{x,r}|$.

    \label{def:delta_ps}
\end{definition}

The set $dom_{x,r}$ contains all the candidates that dominate $r$ in the dataset $x$ and the set $ndom_{x,r}$ contains all the candidates that are not dominated by $r$ in the dataset $x$. The set $dom^-_{x,t,r}$ contains the candidates in $dom_{x,r}$ that can become non-dominated by $r$ in a dataset at distance $t$. Likewise, the set $ndom^+_{x,t,r}$ contains the candidates in $ndom_{x,r}$ that can dominate $r$ in a dataset at distance $t$. The function $\delta^{PS}$ returns the maximum number of candidates that can change their domination status in a dataset at distance $t$.

\begin{example}
    Consider a database $x$, a set of candidates $\range = \{a, b, c\}$ and two utility functions $u_1, u_2: \domain \times \range \rightarrow \real$. The utility functions are defined as follows: 
    \begin{align*}
        u_1(x, a) = 1, u_1(x, b) = 3, u_1(x, c) = 5 \\
        u_2(x, a) = 1, u_2(x, b) = 3, u_2(x, c) = 5
    \end{align*}
    The sensitivity functions at distance $a$ are defined as follows:
    \begin{align*}
        \delta^{u_1}(x,0,a) = 0.5, \delta^{u_2}(x,0,a) = 0.5 \\
        \delta^{u_1}(x,0,b) = 1, \delta^{u_2}(x,0,b) = 1 \\
        \delta^{u_1}(x,0,c) = 1.5, \delta^{u_2}(x,0,c) = 1.5
    \end{align*}

    Thus, to compute the sensitivity function $\delta^{PS}(x,0,b)$, we have:
    \begin{align*}
        dom_{x,r} = \{ c \}, & \; ndom_{x,r} = \{ a \} \\
        dom^-_{x,0,b}  = \{ c \}, & \; ndom^+_{x,0,b}  = \emptyset \\
        \delta^{PS}(x,0,b)  =  &|dom^-_{x,0,b}| + |ndom^+_{x,0,b}| = 1
    \end{align*}
\end{example}

\begin{restatable}{theorem}{theoremadmissibledeltapareto}
    Given utility functions $u_1,\cdots,u_m$, admissible sensitivity functions $\delta^{u_1},\cdots,\delta^{u_m}$ and the Pareto score $PS$ defined over the utility function $u_1,\cdots,u_m$. The sensitivity function $\delta^{PS}(x,t,r)$ is admissible.
    \label{theo:admissible_delta_ps}
\end{restatable}

The full proof is deferred to \iftoggle{vldb}{our technical report \cite{ourtechnicalreport}}{the appendix} and, here, we provide a sketch:

\begin{proof}
    (Sketch) The main section of the proof is to show that $\delta^{PS}(x, t+1, r) \geq \delta^{PS}(y, t, r)$, for all $x,y$ such that $d(x,y) \leq 1$ and all $t \geq 1$. Given two neighboring datasets $x,y$, $d(x,y) \leq 1$, a distance $t$ and a candidate $r$, we have:
    \begin{align}
        & \delta^{PS}(x, t+1, r) = |dom^-_{x,t+1,r}| + |ndom^+_{x,t+1,r}| \nonumber \\
        & = |\{ r' \in dom_{x,r} \; | \; \exists i \in [m] \; s.t. \; u_i^{-(t+1)}(x,r') \leq  u_i^{+(t+1)}(x,r) \}| \nonumber \\
        &  + |\{ r' \in ndom_{x,r} \; | \; u_i^{+(t+1)}(x,r')  \geq  u_i^{-(t+1)}(x,r), \forall i \in [m] \} \nonumber \\
        & = |\{ r' \in (dom_{x,r}\cap dom_{y,r}) \; | \exists i \in [m] \; s.t. \nonumber \\
        & \quad \quad \quad \quad \quad \quad \quad u_i^{-(t+1)}(x,r') \leq  u_i^{+(t+1)}(x,r) \}| \nonumber \\
        & + |\{ r' \in (dom_{x,r}\cap ndom_{y,r}) \; | \; \exists i \in [m] \; s.t. \nonumber\\
        & \quad \quad \quad \quad \quad \quad \quad u_i^{-(t+1)}(x,r') \leq  u_i^{+(t+1)}(x,r) \}| \nonumber \\     
        & + |\{ r' \in (ndom_{x,r} \cap dom_{y,r})  \; | \nonumber \\
        & \quad \quad \quad \quad \quad \quad \quad u_i^{+(t+1)}(x,r')  \geq  u_i^{-(t+1)}(x,r), \forall i \in [m] \}| \nonumber \\
        & + |\{ r' \in (ndom_{x,r} \cap ndom_{y,r})  \; | \nonumber \\ 
        & \quad \quad \quad \quad \quad \quad \quad  u_i^{+(t+1)}(x,r')  \geq  u_i^{-(t+1)}(x,r), \forall i \in [m] \}|. \nonumber
    \end{align}

    The last equality is true since $dom_{x,r} = (dom_{x,r}\cap dom_{y,r}) \cup (dom_{x,r}\cap ndom_{y,r})$ as $dom_{y,r} \cup ndom_{y,r} = \range$  which also implies that $ndom_{y,r} = (ndom_{x,r}\cap dom_{y,r}) \cup (ndom_{x,r}\cap dom_{y,r})$. Also, note that $(dom_{x,r}\cap dom_{y,r}) \cap (dom_{x,r}\cap ndom_{y,r}) = \emptyset$ as $dom_{y,r} \cap ndom_{y,r} = \emptyset$ and $(ndom_{x,r}\cap dom_{y,r}) \cap (ndom_{x,r}\cap ndom_{y,r}) = \emptyset$ for the same reason.

    Note that the terms $dom_{x,r}\cap dom_{y,r}$ contains the candidates that are dominated by $r$ in both datasets $x$ and $y$. The terms $dom_{x,r}\cap ndom_{y,r}$ contains the candidates that are dominated by $r$ in $x$ and not dominated by $r$ in $y$. The terms $ndom_{x,r}\cap dom_{y,r}$ contains the candidates that are not dominated by $r$ in $x$ and dominated by $r$ in $y$. The terms $ndom_{x,r}\cap ndom_{y,r}$ contains the candidates that are not dominated by $r$ in both datasets $x$ and $y$. After that, we find the lower bound for these terms to prove the initial claim. For completeness, one also need show that $\delta^{PS}(x, 0, r) \geq LS^{PS}(y, 0, r)$.
\end{proof}

\section{PrivAgg Mechanism}
\label{sec:privagg}

We present the PrivAgg mechanism, a multi-objective differentially private algorithm for selecting the best candidate from a set of solutions that maximize the aggregate of the multiple objective functions, denoted as the utility functions $u_1, \cdots , u_m$. 

In this approach, we aggregate the utility functions into a single utility function $u_{agg}(x,r)$. More specifically, we define the aggregate utility function $u_{agg}(x,r)$ as the weighted sum of the utility functions $u_{agg} = \sum_{i=1}^{m} w_i \cdot u_i(x,r)$, where $w_i$ is the weight of the utility function $u_i$. The weights can be defined by the data querier. 

Therefore, the PrivAgg mechanism outputs the candidate $r$ that maximizes the aggregate utility function $u_{agg}(x,r)$. Like the PrivPareto mechanism (Section \ref{sec:privpareto}), it has two versions: a global sensitivity version and a local sensitivity version. The global sensitivity version is used when only the global sensitivities of the underlying utility functions $u_1, \cdots , u_m$ are known. Otherwise, the local sensitivity version is used. The global sensitivity version of the PrivAgg mechanism is defined as follows:

\begin{definition}
    (PrivAgg Mechanism - Global Sensitivity). The global PrivAgg mechanism $PrivAgg_{global} (x, \epsilon, u_1, \ldots, u_m, \mathcal{R})$ takes as input a database $x \in \domain$, the privacy budget $\epsilon$, the utility functions $u_1, u_2, \ldots, u_m: \domain \times \range \rightarrow \real$, the set of candidate solutions $\range$, and outputs a candidate $r \in \range$. The mechanism is defined as follows:    
    $$ PrivAgg_{global}(x, \epsilon, u_1, \ldots, u_m, \mathcal{R}) = \mec_{global}(x, \epsilon, u_{agg},\Delta {u_agg}, \mathcal{R}) $$
    where $u_{agg}(x,r) = \sum_{i=1}^{m} w_i \cdot u_i(x,r)$ is the aggregate utility function, $\Delta u_{agg}$ is the global sensitivity of the aggregate utility function, and $\mec_{global}$ is a single objective differentially private mechanism.
\end{definition}

The local sensitivity version of the PrivAgg mechanism is defined as follows:

\begin{definition}
    (PrivAgg Mechanism - Local Sensitivity). The local PrivAgg mechanism $PrivAgg_{local} (x, \epsilon, u_1, \ldots, u_m, \delta^{u_1}, \ldots, \delta^{u_m}, \mathcal{R})$ takes as input a database $x \in \domain$, the privacy budget $\epsilon$, the set of utility functions $u_1, u_2, \ldots, u_m: \domain \times \range \rightarrow \real$, the admissible sensitivity functions $\delta^{u_1} , \ldots, \delta^{u_m}$ of the utility functions, a set of candidate solutions $\range$, and outputs a candidate $r \in \range$. The mechanism is defined as follows:
    \begin{multline}
        PrivAgg_{local} (x, \epsilon, u_1, \ldots, u_m, \delta^{u_1}, \ldots, \delta^{u_m}, \mathcal{R})  \\ = \mec_{local}(x, \epsilon, u_{agg}, \delta^{u_{agg}}, \mathcal{R}).
    \end{multline}
\end{definition}


\subsection{Sensitivity Analysis}

In this subsection we show how to compute an upper bound of the global sensitivity as a function of the underlying utility functions $u_1, \cdots , u_m$ (Theorem \ref{teo:global_sens_agg}). One can safely use the upper bound of the global sensitivity as sensitivity of a global sensitivity based mechanism to ensure $\epsilon$-differential privacy. 

Also, we show how to build an admissible sensitivity function representing the upper bound of the local sensitivity to use on local sensitivity selection algorithms (Theorem \ref{teo:admissible_agg_sens}).

\begin{theorem}
    [Upper Bound on the Global Sensitivity of Aggregate Utility Function]. Let $u_1, \ldots, u_m: \domain \times \range \rightarrow \real$ be utility functions and $w_1, \ldots, w_m$ be weights. We have

    $$ \Delta u_{agg} \leq \sum_{i=1}^{m} |w_i|.\Delta u_i, $$

   \noindent where $\Delta u_i$ is the global sensitivity of the utility function $u_i$.
    \label{teo:global_sens_agg}
\end{theorem}

\begin{proof}
    Let $u_1, \ldots, u_m: \domain \times \range \rightarrow \real$ be utility functions and $w_1, \ldots, w_m$ be weights. We have that

    \begin{align*}
        \Delta u_{agg} & = \max_{x,y \in \domain} \max_{r \in \range} |u_{agg}(x,r) - u_{agg}(y,r)| \\
        & = \max_{x,y \in \domain} \max_{r \in \range} \left| \sum_{i=1}^{m} w_i.u_i(x,r) - \sum_{i=1}^{m} w_i.u_i(y,r) \right| \\
        & = \max_{x,y \in \domain} \max_{r \in \range} \left| \sum_{i=1}^{m} w_i.(u_i(x,r) - u_i(y,r)) \right| \\
        & \leq \sum_{i=1}^{m} |w_i|.\max_{x,y \in \domain} \max_{r \in \range} |u_i(x,r) - u_i(y,r)| = \sum_{i=1}^{m} |w_i|.\Delta u_i.
    \end{align*}
\end{proof}

\begin{theorem} [Admissible Sensitivity Function of Weighted Sum]
    Let $u_1, \ldots, u_m: \domain \times \range \rightarrow \real$ be utility functions, $\delta^{u_1}, \ldots, \delta^{u_m}$ be admissible sensitivity functions, and $w_1, \ldots, w_m$ be weights. Then an admissible sensitivity of the aggregate utility function $u_{agg}(x,r)$ is given by

    $$ \delta^{u_{agg}}(x,t,r) = \sum_{i=1}^{m} |w_i|.\delta^{u_i}(x,t,r), $$

    \noindent where $\delta^{u_i}$ is the admissible sensitivity function of the utility function $u_i$.
    \label{teo:admissible_agg_sens}
\end{theorem}

\begin{proof}

    First, we need to show that $\delta^{u_{agg}}(x,0,r) \geq LS^{u_{agg}}(x,0,r) $, $\forall x \in \domain, r \in \range$. In what follows, we have that:

    \begin{align*}
        \delta^{u_{agg}}(x,0,r) & = \sum_{i=1}^{m} |w_i|.\delta^{u_i}(x,0,r) \geq \sum_{i=1}^{m} |w_i|.LS^{u_i}(x,0,r) \\
        & = \sum_{i=1}^{m} |w_i|.\max_{y \in \domain| d(x,y) \leq 1}|u_i(x,r)-u_i(y,r)| \\
        & \geq  \max_{y \in \domain | d(x,y) \leq 1} |\sum_{i=1}^{m}|w_i|.u_i(x,r)-\sum_{i=1}^{m}|w_i|.u_i(y,r)| \\
        & = LS^{u_{agg}}(x,0,r).
    \end{align*}    

    The first inequality is true as $\delta^{u_i}$ is admissible and the second inequality is due to the triangle inequality. Now, we show that $\delta^{u_{agg}}(x,t+1,r) \geq \delta^{u_{agg}}(y,t,r)$, $\forall x,y \in \domain, r \in \range$ such that $d(x,y) \leq 1$ and $t \geq 0$. Thus, we have that:

    \begin{align*}
        \delta^{u_{agg}}(x,t+1,r) & = \sum_{i=1}^{m} |w_i|.\delta^{u_i}(x,t+1,r) \\
        & \geq \sum_{i=1}^{m} |w_i|.\delta^{u_i}(y,t,r)  = \delta^{u_{agg}}(y,t,r) 
    \end{align*}    
    as $\delta^{u_i}$ is admissible. Therefore, $\delta^{u_{agg}}$ is an admissible sensitivity function of $u_{agg}$.
\end{proof}


\section{Application 1: Cost Sensitive Decision Trees}
\label{section:decision_trees}

We present a differentially private algorithm for building cost-sensitive decision trees. Decision trees are a popular data mining algorithm for classification and regression tasks \cite{kotsiantis2007supervised}. They are easy to interpret and can handle both numerical and categorical data. 

A decision tree is a tree-like structure where each internal node represents a decision based on a feature, each branch represents the outcome of the decision, and each leaf node represents the class label or the value of the target variable. Some algorithms have been proposed to automatically build decision trees using heuristics as the ID3 algorithm \cite{quinlan1986induction}, CART \cite{breiman1984classification}, and C4.5 \cite{quinlan2014c4}. 

In many real problems, the cost of misclassification is different depending on the error type, i.e., false positive or false negative. For example, in medical diagnosis, the cost of a false negative (missing a disease) can be much higher than the cost of a false positive (diagnosing a healthy person as sick). In those cases, we would like to build decision trees that minimize more than one misclassification error. Note that misclassification errors may be conflicting, in the sense that minimizing one error metric may increase another.

\subsection{Problem Statement}

Given a dataset $x$ with $n$ samples and a class label $y_i$ for each sample $x_i$. The task is to build a differentially private cost-sensitive decision tree $\tree$ that minimizes multiple error metrics simultaneously denoted as utility functions $u_1(x, \tree), u_2(x, \tree), \cdots, u_m(x, \tree)$.

In this work, we consider two error metrics as utility functions: True Positive Rate (TPR) and True Negative Rate (TNR). TPR is the ratio of correctly classified positive samples to the total number of positive samples, and TNR is the ratio of correctly classified negative samples to the total number of negative samples:

\begin{align*}
    TPR(x, \tree)=\frac{TP(x, \tree)}{P(x)}, & & TNR(x, \tree)=\frac{TN(x, \tree)}{N(x)},
\end{align*}

\noindent where $TP(x, \tree)$ is the number of true positive predictions of the tree $\tree$ on the dataset $x$, $P(x)$ is the number of positive samples in the dataset $x$, $TN(x, \tree)$ is the number of true negative predictions of the tree $\tree$ on the dataset $x$, and $N(x)$ is the number of negative samples in the dataset $x$.

\subsection{Private Mechanisms}

To tackle this problem, we employ an evolutionary approach to search for the best tree that minimizes the error metrics using a Pareto and an aggregation approach. We first introduce the base algorithm DP-MOET (Differentially Private Multi-Objective Evolutionary Trees). Then, we define the Pareto and aggregation approaches.

\subsubsection{DP-MOET Algorithm}

The DP-MOET algorithm mechanism is an evolutionary algorithm that builds a differentially private cost-sensitive decision trees. This algorithm is inspired by \cite{zhang2013privgene} and \cite{zhao2007multi}. The algorithm is composed of three main steps: initialization, selection, crossover, and mutation. The mechanism is described in Algorithm \ref{algo:dp_moet}.

\begin{algorithm}[htp]    
    \SetKwFunction{algo}{DP-MOET(Database $x$, Privacy Budget $\epsilon$, utility functions $u_1, \cdots, u_m$, Population size $p$, Selection size $s$, Number of iterations $k$, Tree initial depth $d$, Tree Max Depth $d_{max}$, Output Size $o$)}
    \SetKwProg{myalg}{Function}{}{}
    \myalg{\algo}{
	$\TT = random\_tree\_initialization(p, d)$ \\
    \For{$i = 1$ to $k$ }{
    	$\TT = DP\_MOSelection(x,\; \epsilon,\; s,\; u_1, \cdots, u_m,\; \TT )$ \\
        \For{$j = 1$ to $p/2$ }{
            $\tree_1, \tree_2 =$ random selection of two trees from $\TT$  \\
            $\tree_1', \tree_2' = Crossover(\tree_1, \tree_2)$  \\
            $\tree_1' = Mutation(t_1', d_{max})$ \\
            $\tree_2' = Mutation(t_2', d_{max})$ \\
            $\tree_1' = Prune(\tree_1', \; d_{max})$ \\
            $\tree_2' = Prune(\tree_2', \; d_{max})$ \\
            $\TT = \TT \cup \{\tree_1', \tree_2'\}$ \\
        }              
    }     
    \KwRet{$DP\_MOSelection(x,\; \epsilon,\; o,\; u_1, \cdots, u_m,\; \TT )$}\;
	}{}
    \addtocontents{loa}{\protect\addvspace{-9pt}}
    \caption{DP-MOET Algorithm}
    \label{algo:dp_moet}

\end{algorithm}

The algorithm starts by initializing a population of trees $\TT$ with $p$ random trees with depth $d$. Then, it iterates for $k$ generations (Line 2). In each generation, it selects the best $s$ trees from the population $\TT$ using the DP-MOSelection mechanism and stores in $\TT$ itself (Line 4). The DP-MOSelection calls a multi objective selection algorithm (PrivPareto or PrivAgg) $s$ times to select the best $s$ trees from the population $\TT$ on the dataset $x$ with privacy budget $\epsilon$ according to the utility functions $u_1, \cdots, u_m$. Then, it randomly chooses two trees from the population $\TT$ (Line 6) and applies crossover and mutation to generate two new trees (Lines 7-9). The crossover operation may create excessively large trees which can lead to overfitting and, also, increases the search space. Thus, the algorithm prunes the trees to a maximum depth $d_{max}$ (Lines 10 and 11) and adds the new trees to the population $\TT$ (Line 12). Note that the size of the population $\TT$ is $s+p$ after the all crossover and mutation operations in an iteration. Finally, The algorithm returns the best $o$ trees from the population $\TT$ using the DP-MOSelection mechanism (Line 11). These operations are described with more details in the subsection \ref{section:tree_representation}. 

In what follows, we obtain four algorithms from DP-MOET: i) DP-MOET-Pareto\textsubscript{local} by replacing the DP-MOSelection mechanism by PrivPareto\textsubscript{local}; ii) DP-MOET-Pareto\textsubscript{global} by replacing the DP-MOSelection mechanism by PrivPareto\textsubscript{global}; iii) DP-MOET-Agg\textsubscript{local} by replacing the DP-MOSelection mechanism by PrivAgg\textsubscript{local}; and iv) DP-MOET-Agg\textsubscript{global} by replacing the DP-MOSelection mechanism by PrivAgg\textsubscript{global}. The sensitivity analysis for these four mechanisms is provided in the subsection \ref{section:sensitivity_analysis_tree}.

\subsubsection{Tree Representation and Operations}
\label{section:tree_representation}

Now, we provide more details about the tree representation and initialization, crossover, mutation and pruning operations. A decision tree is composed of nodes, where each node is a decision based on an attribute. A node can be a decision node or a terminal node. A terminal node has a class label for prediction. A decision node has an attribute and a threshold to split the data if the attribute is numerical or a value if the attribute is categorical. All numerical attribute are normalized to the range $[0, 1]$, which means that the threshold is also in the range $[0, 1]$. A prediction is made by traversing the tree from the root to the leaf node based on the attribute values of the sample.

The initialization function $random\_tree\_initialization(p, d)$ calls the function $random\_tree(d)$ $p$ times to generate a set of $p$ random trees with depth $d$. The $random\_tree(d)$ function creates a full decision tree where all the leaf nodes are at the same depth $d$ and the decision nodes are randomly selected. The decision nodes are selected by randomly choosing an attribute and a threshold for numerical attributes or a value for categorical attributes. The class label of the leaf nodes is randomly selected from the set of class labels in the dataset.

The crossover operation $Crossover(\tree_1, \tree_2)$ takes two trees $t_1$ and $t_2$ and generates two new trees $\tree_1'$ and $\tree_2'$ by swapping randomly selected nodes, including their subtrees, from the trees $\tree_1$ and $\tree_2$. The mutation operation $Mutation(\tree)$ takes a tree $t$, randomly selects a node and swap to a new randomly create tree using the $random\_tree(d)$ function with $d=d_{max} - d' + 1$ where $d'$ is the depth of the selected node. This prevents the mutation operation to create a tree with a depth greater than $d_{max}$.

The pruning operation $Prune(\tree, d_{max})$ takes a tree $\tree$ and prunes the tree to a maximum depth $d_{max}$ by removing the nodes that are deeper than $d_{max}$ and replacing them by a leaf node with a random class label.

\subsubsection{Sensitivity Analysis}
\label{section:sensitivity_analysis_tree}

\paragraph{DP-MOET-Pareto\textsubscript{local}} The computation of $\delta^{PS}(x,t,r)$ (Definition \ref{def:delta_ps}) requires an admissible sensitivity functions $\delta^{TPR}$ and $\delta^{TNR}$ for the two utility functions $TPR$ and $TNR$, respectively. We define $\delta^{TPR}$ as

\begin{equation*}
    \delta^{TPR}(x,t,\tree) =
    \begin{cases}
        \frac{P'(x, t)-TP'(x, t, \tree)}{P'(x, t)(P'(x, t)-1)} & \text{if } t \leq P(x) -TP(x,\tree) \\
        \frac{TP'(x, t, \tree)}{P'(x, t)(P'(x, t)-1)} & \text{otherwise}
    \end{cases},
\end{equation*}
where 
\begin{align*}
    TP'(x,t,\tree) & = \max\{TP(x,\tree)- \max\{t-(P(x)-TP(x,\tree)), 0\}, 2\}, \\
    P'(x, t) & = \max\{P(x)-t, 2\},
\end{align*}

$P(x)$ is the number of positive samples in the database $x$ and $TP(x,\tree)$ is the number of true positive predictions of the tree $\tree$ on the database $x$. The sensitivity function for $\delta^{TNR}$ is defined similarly by replacing $TP$ by $TN$ and $P$ by $N$ where $TN(x,\tree)$ is the number of true negative predictions of the tree $\tree$ on the database $x$ and $N(x)$ is the number of negative samples in the database $x$. We show the following Lemma:

\begin{restatable}{lemma}{theoremadmissibledeltatprtnr}
    (Admissible Sensitivity Functions). The functions $\delta^{TPR}$ and $\delta^{TNR}$ are admissible sensitivity functions of the utility functions $TPR$ and $TNR$, respectively.
    \label{lemma:admissible_sensitivity_tpr_tnr}
\end{restatable}

The proof of Lemma \ref{lemma:admissible_sensitivity_tpr_tnr} is deferred to \iftoggle{vldb}{our technical report \cite{ourtechnicalreport}}{the appendix}.  Thus, by Theorem \ref{theo:admissible_delta_ps}, the admissible function for the Pareto score $PS$ using $TPR$ and $TNR$, $\delta^{PS}(x,t,r)$, is admissible since $\delta^{TPR}$ and $\delta^{TNR}$ are admissible sensitivity functions of the utility functions $TPR$ and $TNR$, respectively.

\paragraph{DP-MOET-Pareto\textsubscript{global}} As described in Section \ref{sec:sensitivity_pareto_score}, the global sensitivity of the Pareto score used in PrivPareto\textsubscript{global} is $\TT-1$ where $\TT$ is the number of trees in the population.

\paragraph{DP-MOET-Agg\textsubscript{local}} According to Theorem \ref{teo:admissible_agg_sens}, a admissible sensitivity function of the aggregate utility function $u_{agg}$ used in PrivAgg\textsubscript{local} is given by $\delta^{u_{agg}}(x,t,r) = w_{tpr} \cdot \delta^{TPR}(x,t,r) + w_{tnr} \cdot \delta^{TNR}(x,t,r)$ where $w_{tpr}$ and $w_{tnr}$ are the weights of the utility functions $TPR$ and $TNR$, respectively.

\paragraph{DP-MOET-Agg\textsubscript{global}} For the PrivAgg mechanism in DP-MOET-Agg\textsubscript{global} algorithm, we use an upper bound of the global sensitivity given by Theorem \ref{teo:global_sens_agg}, i.e., we use $w_{tpr} \cdot \Delta TPR + w_{tnr} \cdot \Delta TNR$.

In what follows, we need to compute $\Delta TPR$ and $\Delta TNR$. We have that $\Delta TPR = 1$ since the range of the TPR is $[0, 1]$ and we show an example where the maximum and minimum values are obtained in neighboring databases: 

\begin{example}
    let $x$ be a database containing the one single attribute element $x_1=0.7$  associated with output $y_1=1$, let $x'$ be a neighboring database with one element $x'_1=0.3$ associated with output $y'_1=1$ and $\tree$ as a tree composed by a root decision node with threshold $0.5$ and two children leaf nodes with class labels $1$, if value is larger or equal than $0.5$, and $0$ otherwise. Thus $P(x)=P(x')=1$, $TP(x, \tree)=1$ and $TP(x', \tree)=0$ which implies that $TPR(x, \tree)=1$ and $TPR(x', \tree)=0$. 
    \label{example:global_sensitivity_tpr_tnr}
\end{example}

Therefore, $\Delta TPR = 1$. Similarly, we have that $\Delta TNR = 1$. Thus, we use $w_{tpr} + w_{tnr}$ as global sensitivity.

\subsection{Experimental Evaluation - Pareto Approaches}
\label{sec:exp_pareto_decision_tree}

\paragraph{Datasets} i) \textit{National Long Term Care Survey (NLTCS)} \cite{manton2010national} dataset containing $16$ binary attributes with $21,574$ individuals that participated in the survey; ii) \textit{Adult} dataset \cite{blake1998uci} consisting of $45,222$ records with $4$ discrete attributes and $8$ are discrete attributes and; iii) American Community Survey Public Use Microdata Sample dataset (ACS) \cite{series2015version} that gathers general information of $47,461$ people with $23$ binary attributes obtained from 2013 and 2014
    
\paragraph{Methods} We conduct experiments with the two proposed methdods: i) \textit{DP-MOET-Pareto\textsubscript{local}} and ii) \textit{DP-MOET-Pareto\textsubscript{global}}.

\paragraph{Setup} We set the following parameters for experimentation: privacy budget $\epsilon = \{ 0.01, 0.05, 0.1, 0.5, 1.0, 2.0, 5.0 \}$, population size $p=30$, selection size $s=3$, number of iterations $k=3$, tree initial depth $d=4$, tree max depth $d_{max}=10$ and output size $o=3$. As the error metric, we use the metric $C$ used in \cite{zitzler2000comparison}:

\begin{definition}
    Let $X'$, $X''$ be two set of candidates. The function $C$ maps the ordered pair $(X', X'')$ to the interval $[0,1]$:
    
    $$ C(X', X'') = \frac{|\{  a'' \in X'' \; | \; \exists a' \in X' \; s.t. \; a' \succeq a''  \}|}{|X''|} $$
\end{definition}

The metric $C$ measures the percentage of candidates in $X''$ that are dominated by at least one candidate in $X'$. Thus, when $C(X', X'') = 1$, it means that all candidates in $X''$ are dominated by at least one candidate in $X'$ and when $C(X', X'') = 0$, it means that no candidate in $X''$ is dominated by any candidate in $X'$. Therefore, we carry out experiments of the two algorithms DP-MOET-Pareto\textsubscript{local} and DP-MOET-Pareto\textsubscript{global} against a non private version of DP-MOET-Pareto, the \textit{NODP-MOET-Pareto}. NODP-MOET-Pareto is the non-private version of the algorithm obtained replacing the DP-MOSelection in the DP-MOET algorithm (Algorithm \ref{algo:dp_moet}) by the maximum function over the Pareto score.

For each method $M \in \{$ DP-MOET-Pareto\textsubscript{local}, DP-MOET-Pareto\textsubscript{global} $\}$, we run $500$ experiments where, in each experiment, we run and obtain a candidate set $X''$ from $M$, we run and obtain a candidate set $X'$ from NODP-MOET-Pareto and compute $C(X', X'')$. We report the mean value of $C(X', X'')$ over the $500$ runs for each method, dataset and privacy budget. 

\paragraph{Results} Table \ref{table:private_tree_pareto} exhibits the results. As $\epsilon$ grows, the error gets to a plateau for each method dataset. For the NLTCS datset, the error converges to $0.10$, for the Adult dataset, the error converges to $0.08$ and for the ACS dataset, the error converges to $0.10$. DP-MOET-Pareto\textsubscript{local} is shown to be accurate as it gets to the near-minimum error for small privacy budget values ($\epsilon \leq 0.5$). For the DP-MOET-Pareto\textsubscript{global} method, we also tested with very large values of $\epsilon$ to check with what privacy budget the error converges to the minimum. The DP-MOET-Pareto\textsubscript{global} reached the plateau for all datasets with $\epsilon > 10^5$.

\begin{table}[ht]
    \caption{Mean error metric $C$ computed over $500$ runs for DP-MOET-Pareto\textsubscript{global} (MO\textsubscript{g}) and DP-MOET-Pareto\textsubscript{local} (MO\textsubscript{l}) against NODP-MOET-Pareto.}
    \centering
    \begin{tabular}{c|c|c|c|c|c|c|}
        \cline{2-7}
        & \multicolumn{2}{c|}{NLTCS} &  \multicolumn{2}{c|}{Adult} & \multicolumn{2}{c|}{ACS} \\ \hline
        \multicolumn{1}{|c|}{$\epsilon$} & MO\textsubscript{g} & MO\textsubscript{l} & MO\textsubscript{g} & MO\textsubscript{l} & MO\textsubscript{g} & MO\textsubscript{l} \\ \hline
        \multicolumn{1}{|c|}{$0.01$} & 0.48 & 0.45 & 0.21 & 0.18 & 0.28 & 0.25 \\ \hline
        \multicolumn{1}{|c|}{$0.05$} & 0.48 & 0.35 & 0.20 & 0.14 & 0.28 & 0.19 \\ \hline
        \multicolumn{1}{|c|}{$0.1$}  & 0.48 & 0.25 & 0.22 & 0.10 & 0.28 & 0.15 \\ \hline
        \multicolumn{1}{|c|}{$0.5$}  & 0.48 & 0.13 & 0.19 & 0.07 & 0.28 & 0.12 \\ \hline
        \multicolumn{1}{|c|}{$1.0$}  & 0.48 & 0.12 & 0.19 & 0.08 & 0.28 & 0.11 \\ \hline
        \multicolumn{1}{|c|}{$2.0$}  & 0.48 & 0.11 & 0.21 & 0.07 & 0.28 & 0.10 \\ \hline
        \multicolumn{1}{|c|}{$5.0$}  & 0.48 & 0.10 & 0.20 & 0.08 & 0.28 & 0.10 \\ \hline
    \end{tabular}    
    \label{table:private_tree_pareto}

\end{table}

\subsection{Experimental Evaluation - Aggregation Approaches}

We use the same datasets of the previous experimentation on Section \ref{sec:exp_pareto_decision_tree}. 

\paragraph{Methods} We conduct experiments with the following methods: i) DP-MOET-Agg\textsubscript{local}; ii) DP-MOET-Agg\textsubscript{global}; iii) NODP-MOET-Agg, which is the non-private version of the algorithm DP-MOET obtained by replacing the DP-MOSelection mechanism by the maximum over the aggregate utility function and; iv) LocalDiffPID3, the state of the art ID3 differentially private algorithm \cite{farias2023local}.


\paragraph{Evaluation} We set the following parameters for experimentation: privacy budget $\epsilon = \{ 0.01, 0.05, 0.1, 0.5, 1.0, 2.0 \}$, population size $p=100$, selection size $s=20$, number of iterations $k=2$, tree initial depth $d=3$, tree max depth $d_{max}=5$ and output size $o=1$. As error metric, we report the aggregate utility function $u_{agg}$ defined as $u_{agg} = w_{tpr} \cdot TPR + w_{tnr} \cdot TNR$ where $w_{tpr} = 3$ and  $w_{tnr} = 2$.

\paragraph{Results} The experimental results (Figure \ref{fig:weighted_tree}) show that the DP-MOET-Agg\textsubscript{local} achieves a higher fitness than the other private methods and, also, it obtain the same fitness as the non-private method NODP-MOET-Agg for a low privacy budget, $\epsilon = 0.5$. Also, it improves on the LocalDiffPID3 by $21\%$ for adult dataset and $6\%$ for NLTCS dataset. The exception is the ACS dataset where the LocalDiffPID3 outperforms the other methods when $\epsilon \geq 0.5$. However, in the ACS dataset, the DP-MOET-Agg\textsubscript{local} performs better than the LocalDiffPID3 and the other private methods for very small values of privacy budget, $\epsilon < 0.5$.

\pgfplotsset{
    cycle list={%
        {color=brown,mark=triangle,solid},
        {color=red, mark=square,solid},
        {color=blue,mark=o,solid},
        {color=black,dashed}        
    }
}

\begin{figure*}[h!]
    \center
    \begin{tikzpicture}
        \begin{customlegend}[legend columns=4,legend style={align=center,draw=none,column sep=2ex},legend entries={
            DP-MOET-Agg\textsubscript{local}, 
            DP-MOET-Agg\textsubscript{global}, 
            LocalDiffPID3, 
            NODP-MOET-Agg          
        }]
        \addlegendimage{color=brown,mark=triangle,solid}
        \addlegendimage{color=red, mark=square,solid}
        \addlegendimage{color=blue,mark=o,solid}
        \addlegendimage{color=black,dashed}
        \end{customlegend}
    \end{tikzpicture}

\begin{subfigure}[b]{0.32\textwidth}
    \centering
    \resizebox{\linewidth}{!}{
        \begin{tikzpicture}
            \begin{axis}[grid=major,                
                xticklabels={0.01, 0.05, 0.1, 0.5, 1.0, 2.0},
                xtick={1,2,3,4,5,6},
                xlabel={$\epsilon$},
                ylabel={Fitness},                
                xticklabel style={
                    /pgf/number format/fixed,
                    /pgf/number format/precision=5,
                }]

            \addplot coordinates {
                (1, 2.904563508099393)
                (2, 3.708977172901215)
                (3, 3.954701488356122)
                (4, 4.172720607549183)
                (5, 4.187914268881361)
                (6, 4.192291636890738)};  

            \addplot coordinates {
                (1, 2.492106148851644)
                (2, 2.5462128138148423)
                (3, 2.522649260320534)
                (4, 2.555821743738917)
                (5, 2.550482838963625)
                (6, 2.526499059752713)};

            \addplot coordinates {
                (1, 2.956726876137071)
                (2, 3.370477110370742)
                (3, 3.5709581247041666)
                (4, 3.957631625418502)
                (5, 3.9489835914968046)
                (6, 4.014970212347597)};

            \addplot coordinates {
                (1, 4.199937)
                (2, 4.199937)
                (3, 4.199937)
                (4, 4.199937)
                (5, 4.199937)
                (6, 4.199937)};                    

            \end{axis}
        \end{tikzpicture}            
    }
    \caption{NLTCS dataset}
    \label{fig:wtree_nltcs}
\end{subfigure}
\begin{subfigure}[b]{0.32\textwidth}
    \centering
    \resizebox{\linewidth}{!}{
        \begin{tikzpicture}
            \begin{axis}[grid=major,                
                xticklabels={0.01, 0.05, 0.1, 0.5, 1.0, 2.0},
                xtick={1,2,3,4,5,6},
                xlabel={$\epsilon$},
                ylabel={Fitness},                
                xticklabel style={
                    /pgf/number format/fixed,
                    /pgf/number format/precision=5,
                }]
                \addplot coordinates {
                    (1, 2.9281605347302686)
                    (2, 3.5642185743471586)
                    (3, 3.77401986495592)
                    (4, 3.854696024177308)
                    (5, 3.859740405098789)
                    (6, 3.86268201932102)};  
    
                \addplot coordinates {
                    (1, 2.4884417411861173)
                    (2, 2.5142241852740383)
                    (3, 2.510138247819028)
                    (4, 2.497261586299972)
                    (5, 2.557552455153265)
                    (6, 2.536525923727093)};    
    
                \addplot coordinates {
                    (1, 2.3858526483433895)
                    (2, 2.5791783550024983)
                    (3, 3.071014136204832)
                    (4, 3.2892856887697084)
                    (5, 3.387169698853894)
                    (6, 3.398219358663077)};
    
                \addplot coordinates {
                    (1, 3.864989)
                    (2, 3.864989)
                    (3, 3.864989)
                    (4, 3.864989)
                    (5, 3.864989)
                    (6, 3.864989)};   
            \end{axis}
        \end{tikzpicture}            
    }
    \caption{Adult dataset}
    \label{fig:wtree_adult}
\end{subfigure}
\begin{subfigure}[b]{0.32\textwidth}
    \centering
    \resizebox{\linewidth}{!}{
        \begin{tikzpicture}
            \begin{axis}[grid=major,                
                xticklabels={0.01, 0.05, 0.1, 0.5, 1.0, 2.0},
                xtick={1,2,3,4,5,6},
                xlabel={$\epsilon$},
                ylabel={Fitness},                
                xticklabel style={
                    /pgf/number format/fixed,
                    /pgf/number format/precision=5,
                }]
                \addplot coordinates {
                    (1, 3.0678902276444986)
                    (2, 3.868234361008169)
                    (3, 4.045547273877922)
                    (4, 4.210778659478621)
                    (5, 4.226580302361582)
                    (6, 4.2130795049459575)};  
    
                \addplot coordinates {
                    (1, 2.4873258363523933)
                    (2, 2.501130762847372)
                    (3, 2.4891745405813746)
                    (4, 2.5176522780702086)
                    (5, 2.51833286298526)
                    (6, 2.492720749097611)};

                \addplot coordinates {
                    (1, 2.6549005736834044)
                    (2, 3.1924074397522775)
                    (3, 3.635960444864086)
                    (4, 4.6951726963668285)
                    (5, 4.694993224982661)
                    (6, 4.698549494242998)};
    
                \addplot coordinates {
                    (1, 4.239916)
                    (2, 4.239916)
                    (3, 4.239916)
                    (4, 4.239916)
                    (5, 4.239916)
                    (6, 4.239916)};
            \end{axis}
        \end{tikzpicture}            
    }
    \caption{ACS dataset}
    \label{fig:wtree_acs}
\end{subfigure}
\caption{Mean fitness/utility score $u_{agg} = w_{tpr} \cdot TPR + w_{tnr} \cdot TNR$ for $w_{tpr} = 3$, $w_{tnr} = 2$ and $\epsilon \in \{0.01, 0.05, 0.1, 0.5, 1.0, 2.0\}$.}
\label{fig:weighted_tree}
\end{figure*}
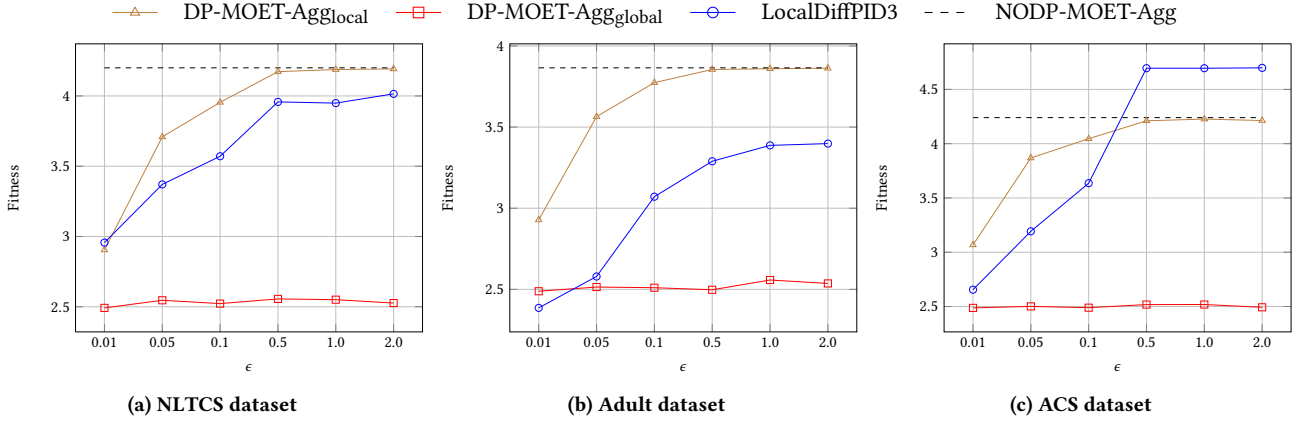

\section{Application 2: Multi Objective Top-k Influential Node Selection}
\label{section:influential}


In social network analysis, detecting influential nodes is a fundamental task that helps determine individuals or entities with the greatest potential to drive information flow, connect communities, or influence interactions within the network \cite{peng2018influence}. However, such analyses often use sensitive data, raising privacy concerns \cite{jiang2021applications}. Differentially private methods address this by adding noise to protect individual privacy while identifying influential nodes \cite{brito2024differentially}.

\subsection{Problem statement}

Given a graph $G=(V,E)$, the multi objective top-k most influential node selection is the problem of selecting $k$ nodes from $V$ that maximizes multiple influence/centrality metrics simultaneously using \textit{edge differential privacy}. The edge differential privacy model is defined as follows:

\begin{definition}[Edge Differential Privacy]
    A randomized algorithm $\mathcal{M}$ satisfies $\epsilon$-edge differential privacy if for any two neighboring graphs $G$ and $G'$ that differ by exactly one edge, and for any possible output $O$ of $\mathcal{M}$, the following holds:
    \[
    \Pr\left[\mathcal{A}(G) \in S\right] \leq e^{\epsilon} \cdot \Pr\left[\mathcal{A}(G') \in S\right],
    \]
    where $\epsilon > 0$ is the privacy budget.
\end{definition}

In this work, as an example, we use two metrics/utility functions: degree centrality and egocentric density \cite{marsden1993reliability}. The degree centrality of a node $v$ is the number of edges incident to the node $degree(G, v)=|N^G(v)|$ where $N^G(v)$ is the set of neighbors of node $v$ in $G$. The egocentric density of a node is the ratio of the number of edges between the neighbors of the node and the maximum number of edges that could exist between them. Formally, the egocentric density of a node $v$ is defined as:

\begin{equation}
    \text{egodensity}(G,v) = \frac{2.|\{(u,w) \in E(G): u,w \in N^G(v)\}|}{|N^G(v)| \times (|N^G(v)|-1)}.
\end{equation}

\subsection{Private Mechanism}

We provide two main strategies to solve this problem: a Pareto based approach and a aggregation based approach. We first provide a base private algorithm that computes a set of $k$ nodes of $V$ that maximizes the two utility functions. Then our proposed algorithms are based on this base algorithm.

Our base algorithm computes a set of $k$ nodes of $V$ that maximizes the two utility functions by calling a differentially private multi objective selection algorithm $k$ times to get one node at each call without replacement. 

From the base algorithm we obtain a family of multi-objective algorithms called \textit{DP-MOTkIN} (Differentially Private Multi-Objective Top-k Influential Node selection). We propose two Pareto based approaches: i) \textit{DP-MOTkIN-Pareto\textsubscript{local}} which is the based algorithm with PrivPareto\textsubscript{local} as the multi objective selection algorithm; and ii) \textit{DP-MOTkIN-Pareto\textsubscript{global}} using PrivPareto\textsubscript{global} likewise. 

Also, we obtain two aggregation approaches based on the base algorithms: i) \textit{DP-MOTkIN-Agg\textsubscript{local}} using PrivAgg\textsubscript{local}; and ii) \textit{DP-MOTkIN-Agg\textsubscript{global}} using PrivAgg\textsubscript{global}.

\subsection{Sensitivity Analysis}

Similar to the analysis done in Section \ref{section:sensitivity_analysis_tree}, we need to compute the global and local sensitivity for degree centrality and egocentric density for the PrivPareto and PrivAgg algorithms in the DP-MOTkIN algorithms.

\paragraph{Degree} The global sensitivity of the degree is $\Delta degree = 1$ and the element local sensitivity $LS^{degree}(G,t,v)=1$, because for any neighboring database that differs by one edge, say $e = (u,v)$, the degrees of $u$ and $v$ each change by exactly one.

\paragraph{Egocentric density} The global sensitivity of the egocentric density is $\Delta degree = 1$, which is set by the following example: let $G=(V,E)$ be a graph with $V=\{v_1,v_2,v_3\}$ and $E=\{(v_1,v_2),(v_1,v_3)\}$. Let $G'=(V',E')$ be a graph with $V'=\{v_1,v_2,v_3\}$ and $E'=\{(v_1,v_2),(v_2,v_3), (v_1, v_3)\}$. In this example, $egodensity(G, v)=0$ and $egodensity(G', v)=1$, which are the minimum and the maximum of the range of $egodensity$, respectively. An admissible sensitivity function for egocentric density is given as:

\begin{equation*}
    \delta^{egodensity}(G,t,v) =
    \begin{cases}
        2/(|N^G(v)|-t-2) & \text{if} \ |N^G(v)|-t > 2 \\
        1 & \text{otherwise}
    \end{cases}, 
\end{equation*}

Next, we show that the sensitivity function $\delta^{egodensity}$ is admissible. The proof is deferred to \iftoggle{vldb}{our technical report \cite{ourtechnicalreport}}{the appendix}.

\begin{restatable}{lemma}{theoremadmissibledeltaegodensity}
    The sensitivity function $\delta^{egodensity}(x,t,r)$ is admissible for $egodensity$.
    \label{lemma:egodensity_sensitivity}
\end{restatable}

\subsection{Experimental Evaluation - Pareto Approaches}
\label{sec:exp_pareto_graph}

\pgfplotsset{
    cycle list={%
        {color=brown,mark=triangle,solid},
        {color=red, mark=square,solid},
    }
}

\begin{figure*}[h!]
    \center
    \begin{tikzpicture}
        \begin{customlegend}[legend columns=2,legend style={align=center,draw=none,column sep=2ex},legend entries={
            DP-MOTkIN-Agg\textsubscript{global},
            DP-MOTkIN-Agg\textsubscript{local},
        }]
        \addlegendimage{color=brown,mark=triangle,solid}
        \addlegendimage{color=red, mark=square,solid}
        \end{customlegend}
    \end{tikzpicture}

\begin{subfigure}[b]{0.32\textwidth}
    \centering
    \resizebox{\linewidth}{!}{
        \begin{tikzpicture}
            \begin{axis}[grid=major,                
                xticklabels={$0.01$, $0.05$, $0.1$, $0.5$, $1.0$, $2.0$, $5.0$, $10.0$, $20.0$, $100.0$},
                xtick={1,2,3,4,5,6,7,8,9,10},
                xlabel={$\epsilon$},
                ylabel={Recall},                
                xticklabel style={
                    /pgf/number format/fixed,
                    /pgf/number format/precision=5,
                }]

            \addplot coordinates {
                (1, 0.0)
                (2, 0.0004)
                (3, 0.0004)
                (4, 0.0004)
                (5, 0.002)
                (6, 0.022)
                (7, 0.7732)
                (8, 0.9368000000000001)
                (9, 0.9936)
                (10, 1.0)};  

            \addplot coordinates {
                (1, 0.0004)
                (2, 0.0032)
                (3, 0.0884)
                (4, 0.9656)
                (5, 0.998)
                (6, 1.0)
                (7, 1.0)
                (8, 1.0)
                (9, 1.0)
                (10, 1.0)};

            \end{axis}
        \end{tikzpicture}            
    }
    \caption{Enron dataset}
\end{subfigure}
\begin{subfigure}[b]{0.32\textwidth}
    \centering
    \resizebox{\linewidth}{!}{
        \begin{tikzpicture}
            \begin{axis}[grid=major,                
                xticklabels={$0.01$, $0.05$, $0.1$, $0.5$, $1.0$, $2.0$, $5.0$, $10.0$, $20.0$, $100.0$},
                xtick={1,2,3,4,5,6,7,8,9,10},
                xlabel={$\epsilon$},
                ylabel={Recall},                
                xticklabel style={
                    /pgf/number format/fixed,
                    /pgf/number format/precision=5,
                }]
                \addplot coordinates {
                    (1, 0.0)
                    (2, 0.0)
                    (3, 0.0)
                    (4, 0.0)
                    (5, 0.0)
                    (6, 0.0)
                    (7, 0.0004)
                    (8, 0.0016)
                    (9, 0.17360000000000003)
                    (10, 0.9972000000000001)};  
                
                \addplot coordinates {
                    (1, 0.0)
                    (2, 0.0)
                    (3, 0.0)
                    (4, 0.0244)
                    (5, 0.7196)
                    (6, 0.95)
                    (7, 0.9992000000000001)
                    (8, 1.0)
                    (9, 1.0)
                    (10, 1.0)}; 
    
            \end{axis}
        \end{tikzpicture}            
    }
    \caption{DBLP dataset}
\end{subfigure}
\begin{subfigure}[b]{0.32\textwidth}
    \centering
    \resizebox{\linewidth}{!}{
        \begin{tikzpicture}
            \begin{axis}[grid=major,                
                xticklabels={$0.01$, $0.05$, $0.1$, $0.5$, $1.0$, $2.0$, $5.0$, $10.0$, $20.0$, $100.0$},
                xtick={1,2,3,4,5,6,7,8,9,10},
                xlabel={$\epsilon$},
                ylabel={Recall},                
                xticklabel style={
                    /pgf/number format/fixed,
                    /pgf/number format/precision=5,
                }]

                \addplot coordinates {
                    (1, 0.0004)
                    (2, 0.0)
                    (3, 0.0004)
                    (4, 0.1728)
                    (5, 0.41520000000000007)
                    (6, 0.8132)
                    (7, 0.9528)
                    (8, 0.9836)
                    (9, 0.9992000000000001)
                    (10, 1.0)};

                \addplot coordinates {
                    (1, 0.005600000000000001)
                    (2, 0.5008)
                    (3, 0.912)
                    (4, 0.9896)
                    (5, 0.9996)
                    (6, 1.0)
                    (7, 1.0)
                    (8, 1.0)
                    (9, 1.0)
                    (10, 1.0)};
                        
            \end{axis}
        \end{tikzpicture}  

    }
    \caption{Github dataset}
\end{subfigure}
\caption{Mean recall for DP-MOTkIN-Agg methods where $u_{agg} = w_{degree} \cdot degree + w_{egodensity} \cdot egodensity$ for $w_{degree} = 1$ and $w_{egodensity} = 100$ for $\epsilon \in \{0.01, 0.05, 0.1, 0.5, 1.0, 2.0\}$ and $k=5$.}
\label{fig:graph_weighted}
\end{figure*}
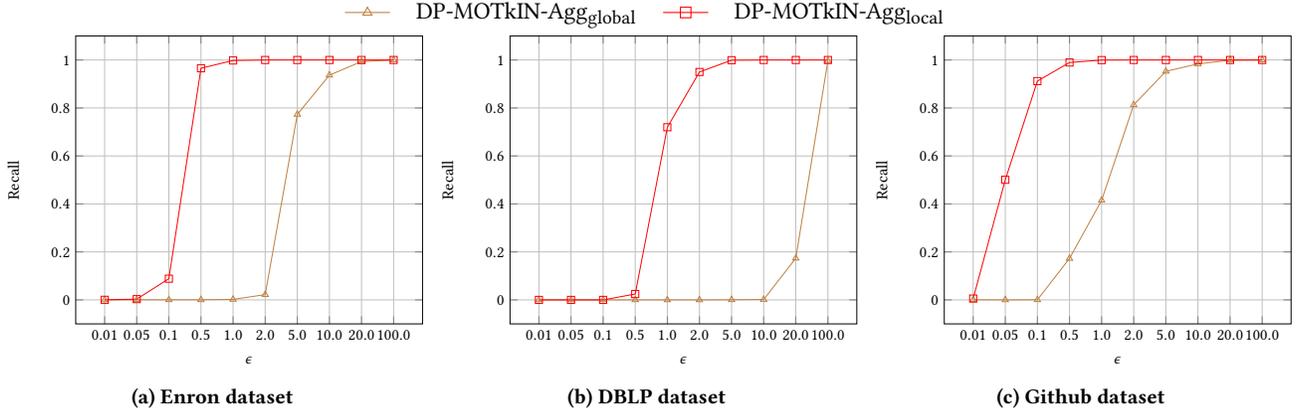


\paragraph{Datasets} We utilize three real-world graph datasets: 1) \textit{Github}: A network of developers with at least $10$ stars on the platform. Nodes represent developers, and edges indicate mutual follows between developers ($|V|=37,700$, $|E|=289,003$, $\Delta(G) = 9,458$). Here, $V$ denotes the set of vertices, $E$ the set of edges, and $\Delta(G)$ the maximum degree of graph $G$. The dataset is available in the Stanford Network Dataset Collection \cite{snapnets}.
2) \textit{Enron}: An email communication network derived from approximately a half-million emails. Nodes correspond to email addresses, and edges connect pairs of addresses that exchanged emails ($|V|=36,692$, $|E|=183,831$, $\Delta(G) = 1,383$).
3) \textit{DBLP}: A co-authorship network where nodes represent authors, and edges indicate co-authorship of at least one paper ($|V|=317,080$, $|E|=1,049,866$, $\Delta(G) = 343$).

\paragraph{Methods} We carry out the experiments with the following methods: 1) \textit{DP-MOTkIN-Pareto\textsubscript{local}}; and 2) \textit{DP-MOTkIN-Pareto\textsubscript{global}}.

\paragraph{Evalutation} Similar to the experimental setup described in Section \ref{sec:exp_pareto_decision_tree}, we report the mean metric $C$ to compare against the true top-k with largest Pareto score over $500$ repetitions. We set the privacy budget $\epsilon=\{0.1, 0.5, 1.0, 2.0, 5.0, 10.0, 20.0, 50.0\}$.  We set $k=3$ for the Github and Enron datasets, and $k=10$ for DBLP since it is a large dataset). 

\paragraph{Results} Table \ref{table:topk_graph_pareto} displays the results of the experiment. DP-MOTkIN-Pareto\textsubscript{local} achieves better results than DP-MOTkIN-Pareto\textsubscript{global} in all datasets for all values of the privacy budget ($\epsilon$). DP-MOTkIN-Pareto\textsubscript{local} achieves zero error with low privacy budget for Github ($\epsilon < 0.5$) since the range of degrees is large compared to the sensitivity of the degree, which entails in low Pareto score sensitivity. It gets the same level of error as DP-MOTkIN-Pareto\textsubscript{global} with two orders of magnitude less privacy budget. For the Enron and DBLP datasets, DP-MOTkIN-Pareto\textsubscript{local} achieves a reasonable error with $\epsilon = 20.0$. This is due to the fact that the range of degrees of Enron and DBLP are not as large as the Github dataset.

\begin{table}[ht]
    \caption{Mean error metric $C$ compared against the true top-k computed over $500$ runs for DP-MOTkIN-Pareto\textsubscript{global} (MOT\textsubscript{g}) and DP-MOTkIN-Pareto\textsubscript{local} (MOT\textsubscript{l}) for  $\epsilon=\{10^{-1},10^{0},10^{1},10^{2},10^{3}\}$ and $k=3$.}
    \centering
    \begin{tabular}{c|c|c|c|c|c|c|}
        \cline{2-7}
        & \multicolumn{2}{c|}{Github} &  \multicolumn{2}{c|}{Enron} & \multicolumn{2}{c|}{DBLP} \\ \hline
        \multicolumn{1}{|c|}{$\epsilon$} & MOT\textsubscript{g} & MOT\textsubscript{l} & MOT\textsubscript{g} & MOT\textsubscript{l} & MOT\textsubscript{g} & MOT\textsubscript{l} \\ \hline
        \multicolumn{1}{|c|}{$0.1$}  & 0.20 & 0.15 & 0.41 & 0.41 & 0.07 & 0.08 \\ \hline
        \multicolumn{1}{|c|}{$0.5$}  & 0.21 & 0.00 & 0.39 & 0.41 & 0.09 & 0.09 \\ \hline
        \multicolumn{1}{|c|}{$1.0$}  & 0.22 & 0.00 & 0.42 & 0.39 & 0.09 & 0.08 \\ \hline
        \multicolumn{1}{|c|}{$2.0$}  & 0.18 & 0.00 & 0.34 & 0.38 & 0.09 & 0.08 \\ \hline
        \multicolumn{1}{|c|}{$5.0$}  & 0.16 & 0.00 & 0.33 & 0.30 & 0.08 & 0.07 \\ \hline
        \multicolumn{1}{|c|}{$10.0$} & 0.10 & 0.00 & 0.22 & 0.16 & 0.07 & 0.06 \\ \hline
        \multicolumn{1}{|c|}{$20.0$} & 0.05 & 0.00 & 0.15 & 0.07 & 0.07 & 0.05 \\ \hline
        \multicolumn{1}{|c|}{$50.0$} & 0.00 & 0.00 & 0.07 & 0.00 & 0.05 & 0.00 \\ \hline
    \end{tabular}    
    \label{table:topk_graph_pareto}
\end{table}


\subsection{Experimental Evaluation - Aggregation Approaches}

We use the same datasets of the previous experiment (Section \ref{sec:exp_pareto_graph}). 

\paragraph{Methods} We experiment with the following methods: 1) \textit{DP-MOTkIN-Agg\textsubscript{local}}; and 2) \textit{DP-MOTkIN-Agg\textsubscript{global}}.

\paragraph{Evaluation} We report the mean recall of the true top-k nodes over $500$ repetitions using the aggregate utility function $u_{agg} = w_{degree} \cdot degree + w_{egodensity} \cdot egodensity$ with $w_{degree}=1$ and $w_{egodensity}=100$.  We set the privacy budget $\epsilon=\{10^{-1},10^{0},10^{1},10^{2},\allowbreak 10^{3}\}$ and $k=5$.

\paragraph{Results} Figure \ref{fig:graph_weighted} displays the results. DP-MOTkIN-Agg\textsubscript{local} obtains perfect recall with low privacy budget, $\epsilon \leq 1$ for Enron and Github datasets and $\epsilon \leq 5$ for DBLP dataset. Moreover, in the DBLP dataset, DP-MOTkIN-Agg\textsubscript{local} achieves reasonable recall ($0.71$) with $\epsilon=1$. On the other hand, the DP-MOTkIN-Agg\textsubscript{global} needs a higher privacy budget, from $1$ to $2$ orders of magnitude, to achieve the same recall as DP-MOTkIN-Agg\textsubscript{local}.

\section{Conclusion}

In this work, we introduced two novel mechanisms for the differentially private multi-objective selection problem: PrivPareto and PrivAgg. These mechanisms address the important challenge of selecting candidates that optimize multiple competing objectives while preserving privacy. The PrivPareto mechanism uses a novel Pareto score to identify solutions close to the Pareto frontier, while PrivAgg enables weighted aggregation of multiple objectives into a single utility function. As sensitivity measure, both PrivaPareto and PrivAgg can use global or local sensitivity. Our theoretical framework provides the means to compute sensitivity bounds based on the underlying utility functions. 


We show the practical applicability of our mechanisms through two real-world applications: i) cost-sensitive decision tree construction where we propose a family of evolutionary algorithms $DP-MOET$ and ii) multi-objective influential node selection in social networks where we propose the family of algorithms $DP-MOTkIN$. 


The experimental results showed that our local sensitivity-based approaches achieve significantly better utility compared to global sensitivity approaches across both applications and both Pareto and Aggregation approaches. Moreover, the local sensitivity-based approaches are able to perform well with typical privacy budget values $\epsilon \in [0.01, 1]$ in most experiments.



Our mechanisms provide a fundamental building block for privacy-preserving multi-objective optimization that can be composed into more complex differentially private algorithms. We believe this work opens up new possibilities for tackling real-world problems that require balancing multiple competing objectives under privacy constraints.

\bibliographystyle{ACM-Reference-Format}
\bibliography{sample}


\begin{thebibliography}{40}


\ifx \showCODEN    \undefined \def \showCODEN     #1{\unskip}     \fi
\ifx \showDOI      \undefined \def \showDOI       #1{#1}\fi
\ifx \showISBNx    \undefined \def \showISBNx     #1{\unskip}     \fi
\ifx \showISBNxiii \undefined \def \showISBNxiii  #1{\unskip}     \fi
\ifx \showISSN     \undefined \def \showISSN      #1{\unskip}     \fi
\ifx \showLCCN     \undefined \def \showLCCN      #1{\unskip}     \fi
\ifx \shownote     \undefined \def \shownote      #1{#1}          \fi
\ifx \showarticletitle \undefined \def \showarticletitle #1{#1}   \fi
\ifx \showURL      \undefined \def \showURL       {\relax}        \fi
\providecommand\bibfield[2]{#2}
\providecommand\bibinfo[2]{#2}
\providecommand\natexlab[1]{#1}
\providecommand\showeprint[2][]{arXiv:#2}

\bibitem[\protect\citeauthoryear{Blake and Merz}{Blake and Merz}{1998}]%
        {blake1998uci}
\bibfield{author}{\bibinfo{person}{Catherine~L Blake} {and} \bibinfo{person}{Christopher~J Merz}.} \bibinfo{year}{1998}\natexlab{}.
\newblock \bibinfo{title}{UCI repository of machine learning databases}.
\newblock
\newblock


\bibitem[\protect\citeauthoryear{Blocki, Blum, Datta, and Sheffet}{Blocki et~al\mbox{.}}{2013}]%
        {blocki2013differentially}
\bibfield{author}{\bibinfo{person}{Jeremiah Blocki}, \bibinfo{person}{Avrim Blum}, \bibinfo{person}{Anupam Datta}, {and} \bibinfo{person}{Or Sheffet}.} \bibinfo{year}{2013}\natexlab{}.
\newblock \showarticletitle{Differentially private data analysis of social networks via restricted sensitivity}. In \bibinfo{booktitle}{\emph{Proceedings of the 4th Conference on Innovations in Theoretical Computer Science}}. ACM, \bibinfo{pages}{87--96}.
\newblock


\bibitem[\protect\citeauthoryear{Breiman, Friedman, Stone, and Olshen}{Breiman et~al\mbox{.}}{1984}]%
        {breiman1984classification}
\bibfield{author}{\bibinfo{person}{Leo Breiman}, \bibinfo{person}{Jerome Friedman}, \bibinfo{person}{Charles~J Stone}, {and} \bibinfo{person}{Richard~A Olshen}.} \bibinfo{year}{1984}\natexlab{}.
\newblock \bibinfo{booktitle}{\emph{Classification and regression trees}}.
\newblock \bibinfo{publisher}{CRC press}.
\newblock


\bibitem[\protect\citeauthoryear{Brito, Mendon{\c{c}}a, and Machado}{Brito et~al\mbox{.}}{2024}]%
        {brito2024differentially}
\bibfield{author}{\bibinfo{person}{Felipe~T Brito}, \bibinfo{person}{Andr{\'e}~LC Mendon{\c{c}}a}, {and} \bibinfo{person}{Javam~C Machado}.} \bibinfo{year}{2024}\natexlab{}.
\newblock \showarticletitle{A Differentially Private Guide for Graph Analytics.}. In \bibinfo{booktitle}{\emph{EDBT}}. \bibinfo{pages}{850--853}.
\newblock


\bibitem[\protect\citeauthoryear{de~Farias, Brito, Flynn, Machado, Majumdar, and Srivastava}{de~Farias et~al\mbox{.}}{2020}]%
        {DBLP:journals/pvldb/FariasBFMMS20}
\bibfield{author}{\bibinfo{person}{Victor A.~E. de Farias}, \bibinfo{person}{Felipe~T. Brito}, \bibinfo{person}{Cheryl Flynn}, \bibinfo{person}{Javam~C. Machado}, \bibinfo{person}{Subhabrata Majumdar}, {and} \bibinfo{person}{Divesh Srivastava}.} \bibinfo{year}{2020}\natexlab{}.
\newblock \showarticletitle{Local Dampening: Differential Privacy for Non-numeric Queries via Local Sensitivity}.
\newblock \bibinfo{journal}{\emph{Proc. {VLDB} Endow.}} \bibinfo{volume}{14}, \bibinfo{number}{4} (\bibinfo{year}{2020}), \bibinfo{pages}{521--533}.
\newblock
\urldef\tempurl%
\url{https://doi.org/10.14778/3436905.3436912}
\showDOI{\tempurl}


\bibitem[\protect\citeauthoryear{Ding, Kifer, E., Steinke, Wang, Xiao, and Zhang}{Ding et~al\mbox{.}}{2021}]%
        {ding2021permute}
\bibfield{author}{\bibinfo{person}{Zeyu Ding}, \bibinfo{person}{Daniel Kifer}, \bibinfo{person}{Sayed M. Saghaian~N. E.}, \bibinfo{person}{Thomas Steinke}, \bibinfo{person}{Yuxin Wang}, \bibinfo{person}{Yingtai Xiao}, {and} \bibinfo{person}{Danfeng Zhang}.} \bibinfo{year}{2021}\natexlab{}.
\newblock \showarticletitle{{T}he {P}ermute-and-Flip {M}echanism is {I}dentical to {R}eport-{N}oisy-Max with {E}xponential {N}oise}.
\newblock \bibinfo{journal}{\emph{CoRR}}  \bibinfo{volume}{abs/2105.07260} (\bibinfo{year}{2021}).
\newblock
\showeprint[arXiv]{2105.07260}
\urldef\tempurl%
\url{https://arxiv.org/abs/2105.07260}
\showURL{%
\tempurl}


\bibitem[\protect\citeauthoryear{Durfee and Rogers}{Durfee and Rogers}{2019}]%
        {durfee2019practical}
\bibfield{author}{\bibinfo{person}{David Durfee} {and} \bibinfo{person}{Ryan~M. Rogers}.} \bibinfo{year}{2019}\natexlab{}.
\newblock \showarticletitle{{P}ractical {D}ifferentially {P}rivate {T}op-k {S}election with {P}ay-what-you-get {C}omposition}. In \bibinfo{booktitle}{\emph{Advances in Neural Information Processing Systems 32: Annual Conference on Neural Information Processing Systems 2019, NeurIPS 2019, December 8-14, 2019, Vancouver, BC, Canada}}, \bibfield{editor}{\bibinfo{person}{Hanna~M. Wallach}, \bibinfo{person}{Hugo Larochelle}, \bibinfo{person}{Alina Beygelzimer}, \bibinfo{person}{Florence d'Alch{\'{e}}{-}Buc}, \bibinfo{person}{Emily~B. Fox}, {and} \bibinfo{person}{Roman Garnett}} (Eds.). \bibinfo{pages}{3527--3537}.
\newblock
\urldef\tempurl%
\url{https://proceedings.neurips.cc/paper/2019/hash/b139e104214a08ae3f2ebcce149cdf6e-Abstract.html}
\showURL{%
\tempurl}


\bibitem[\protect\citeauthoryear{Dwork}{Dwork}{2011}]%
        {dwork2011differential}
\bibfield{author}{\bibinfo{person}{Cynthia Dwork}.} \bibinfo{year}{2011}\natexlab{}.
\newblock \showarticletitle{Differential privacy}.
\newblock \bibinfo{journal}{\emph{Encyclopedia of Cryptography and Security}} (\bibinfo{year}{2011}), \bibinfo{pages}{338--340}.
\newblock


\bibitem[\protect\citeauthoryear{Dwork, Kenthapadi, McSherry, Mironov, and Naor}{Dwork et~al\mbox{.}}{2006a}]%
        {dwork2006our}
\bibfield{author}{\bibinfo{person}{Cynthia Dwork}, \bibinfo{person}{Krishnaram Kenthapadi}, \bibinfo{person}{Frank McSherry}, \bibinfo{person}{Ilya Mironov}, {and} \bibinfo{person}{Moni Naor}.} \bibinfo{year}{2006}\natexlab{a}.
\newblock \showarticletitle{Our data, ourselves: Privacy via distributed noise generation}. In \bibinfo{booktitle}{\emph{Annual International Conference on the Theory and Applications of Cryptographic Techniques}}. Springer, \bibinfo{pages}{486--503}.
\newblock


\bibitem[\protect\citeauthoryear{Dwork, McSherry, Nissim, and Smith}{Dwork et~al\mbox{.}}{2006b}]%
        {dwork2006calibrating}
\bibfield{author}{\bibinfo{person}{Cynthia Dwork}, \bibinfo{person}{Frank McSherry}, \bibinfo{person}{Kobbi Nissim}, {and} \bibinfo{person}{Adam Smith}.} \bibinfo{year}{2006}\natexlab{b}.
\newblock \showarticletitle{Calibrating noise to sensitivity in private data analysis}. In \bibinfo{booktitle}{\emph{Theory of cryptography conference}}. Springer, \bibinfo{pages}{265--284}.
\newblock


\bibitem[\protect\citeauthoryear{Dwork, Roth, et~al\mbox{.}}{Dwork et~al\mbox{.}}{2014}]%
        {dwork2014algorithmic}
\bibfield{author}{\bibinfo{person}{Cynthia Dwork}, \bibinfo{person}{Aaron Roth}, {et~al\mbox{.}}} \bibinfo{year}{2014}\natexlab{}.
\newblock \showarticletitle{The algorithmic foundations of differential privacy.}
\newblock \bibinfo{journal}{\emph{Foundations and Trends in Theoretical Computer Science}} \bibinfo{volume}{9}, \bibinfo{number}{3-4} (\bibinfo{year}{2014}), \bibinfo{pages}{211--407}.
\newblock


\bibitem[\protect\citeauthoryear{Farias, Brito, Flynn, Machado, Majumdar, and Srivastava}{Farias et~al\mbox{.}}{2023}]%
        {farias2023local}
\bibfield{author}{\bibinfo{person}{Victor~AE Farias}, \bibinfo{person}{Felipe~T Brito}, \bibinfo{person}{Cheryl Flynn}, \bibinfo{person}{Javam~C Machado}, \bibinfo{person}{Subhabrata Majumdar}, {and} \bibinfo{person}{Divesh Srivastava}.} \bibinfo{year}{2023}\natexlab{}.
\newblock \showarticletitle{Local dampening: Differential privacy for non-numeric queries via local sensitivity}.
\newblock \bibinfo{journal}{\emph{the VLDB Journal}} \bibinfo{volume}{32}, \bibinfo{number}{6} (\bibinfo{year}{2023}), \bibinfo{pages}{1191--1214}.
\newblock


\bibitem[\protect\citeauthoryear{Fletcher and Islam}{Fletcher and Islam}{2019}]%
        {fletcher2019decision}
\bibfield{author}{\bibinfo{person}{Sam Fletcher} {and} \bibinfo{person}{Md~Zahidul Islam}.} \bibinfo{year}{2019}\natexlab{}.
\newblock \showarticletitle{Decision tree classification with differential privacy: A survey}.
\newblock \bibinfo{journal}{\emph{ACM Computing Surveys (CSUR)}} \bibinfo{volume}{52}, \bibinfo{number}{4} (\bibinfo{year}{2019}), \bibinfo{pages}{1--33}.
\newblock


\bibitem[\protect\citeauthoryear{Hamza, Hefny, et~al\mbox{.}}{Hamza et~al\mbox{.}}{2013}]%
        {hamza2013attacks}
\bibfield{author}{\bibinfo{person}{Nermin Hamza}, \bibinfo{person}{Hesham~A Hefny}, {et~al\mbox{.}}} \bibinfo{year}{2013}\natexlab{}.
\newblock \showarticletitle{Attacks on anonymization-based privacy-preserving: a survey for data mining and data publishing}.
\newblock  (\bibinfo{year}{2013}).
\newblock


\bibitem[\protect\citeauthoryear{Henriksen-Bulmer and Jeary}{Henriksen-Bulmer and Jeary}{2016}]%
        {henriksen2016re}
\bibfield{author}{\bibinfo{person}{Jane Henriksen-Bulmer} {and} \bibinfo{person}{Sheridan Jeary}.} \bibinfo{year}{2016}\natexlab{}.
\newblock \showarticletitle{Re-identification attacks—A systematic literature review}.
\newblock \bibinfo{journal}{\emph{International Journal of Information Management}} \bibinfo{volume}{36}, \bibinfo{number}{6} (\bibinfo{year}{2016}), \bibinfo{pages}{1184--1192}.
\newblock


\bibitem[\protect\citeauthoryear{Ji, Lipton, and Elkan}{Ji et~al\mbox{.}}{2014}]%
        {ji2014differential}
\bibfield{author}{\bibinfo{person}{Zhanglong Ji}, \bibinfo{person}{Zachary~C Lipton}, {and} \bibinfo{person}{Charles Elkan}.} \bibinfo{year}{2014}\natexlab{}.
\newblock \showarticletitle{Differential privacy and machine learning: a survey and review}.
\newblock \bibinfo{journal}{\emph{arXiv preprint arXiv:1412.7584}} (\bibinfo{year}{2014}).
\newblock


\bibitem[\protect\citeauthoryear{Jiang, Pei, Yu, Yu, Gong, and Cheng}{Jiang et~al\mbox{.}}{2021}]%
        {jiang2021applications}
\bibfield{author}{\bibinfo{person}{Honglu Jiang}, \bibinfo{person}{Jian Pei}, \bibinfo{person}{Dongxiao Yu}, \bibinfo{person}{Jiguo Yu}, \bibinfo{person}{Bei Gong}, {and} \bibinfo{person}{Xiuzhen Cheng}.} \bibinfo{year}{2021}\natexlab{}.
\newblock \showarticletitle{Applications of differential privacy in social network analysis: A survey}.
\newblock \bibinfo{journal}{\emph{IEEE transactions on knowledge and data engineering}} \bibinfo{volume}{35}, \bibinfo{number}{1} (\bibinfo{year}{2021}), \bibinfo{pages}{108--127}.
\newblock


\bibitem[\protect\citeauthoryear{Karwa, Raskhodnikova, Smith, and Yaroslavtsev}{Karwa et~al\mbox{.}}{2011}]%
        {karwa2011private}
\bibfield{author}{\bibinfo{person}{Vishesh Karwa}, \bibinfo{person}{Sofya Raskhodnikova}, \bibinfo{person}{Adam Smith}, {and} \bibinfo{person}{Grigory Yaroslavtsev}.} \bibinfo{year}{2011}\natexlab{}.
\newblock \showarticletitle{Private analysis of graph structure}.
\newblock \bibinfo{journal}{\emph{PVLDB}} \bibinfo{volume}{4}, \bibinfo{number}{11} (\bibinfo{year}{2011}), \bibinfo{pages}{1146--1157}.
\newblock


\bibitem[\protect\citeauthoryear{Kasiviswanathan, Nissim, Raskhodnikova, and Smith}{Kasiviswanathan et~al\mbox{.}}{2013}]%
        {kasiviswanathan2013analyzing}
\bibfield{author}{\bibinfo{person}{Shiva~Prasad Kasiviswanathan}, \bibinfo{person}{Kobbi Nissim}, \bibinfo{person}{Sofya Raskhodnikova}, {and} \bibinfo{person}{Adam Smith}.} \bibinfo{year}{2013}\natexlab{}.
\newblock \showarticletitle{Analyzing graphs with node differential privacy}. In \bibinfo{booktitle}{\emph{Theory of Cryptography Conference}}. Springer, \bibinfo{pages}{457--476}.
\newblock


\bibitem[\protect\citeauthoryear{Kotsiantis, Zaharakis, and Pintelas}{Kotsiantis et~al\mbox{.}}{2007}]%
        {kotsiantis2007supervised}
\bibfield{author}{\bibinfo{person}{Sotiris~B Kotsiantis}, \bibinfo{person}{I Zaharakis}, {and} \bibinfo{person}{P Pintelas}.} \bibinfo{year}{2007}\natexlab{}.
\newblock \showarticletitle{Supervised machine learning: A review of classification techniques}.
\newblock \bibinfo{journal}{\emph{Emerging artificial intelligence applications in computer engineering}}  \bibinfo{volume}{160} (\bibinfo{year}{2007}), \bibinfo{pages}{3--24}.
\newblock


\bibitem[\protect\citeauthoryear{Leskovec and Krevl}{Leskovec and Krevl}{2014}]%
        {snapnets}
\bibfield{author}{\bibinfo{person}{Jure Leskovec} {and} \bibinfo{person}{Andrej Krevl}.} \bibinfo{year}{2014}\natexlab{}.
\newblock \bibinfo{title}{{SNAP Datasets}: {Stanford} Large Network Dataset Collection}.
\newblock \bibinfo{howpublished}{http://snap.stanford.edu/data}.
\newblock


\bibitem[\protect\citeauthoryear{Liu, Ding, Shaham, Rahayu, Farokhi, and Lin}{Liu et~al\mbox{.}}{2021}]%
        {liu2021machine}
\bibfield{author}{\bibinfo{person}{Bo Liu}, \bibinfo{person}{Ming Ding}, \bibinfo{person}{Sina Shaham}, \bibinfo{person}{Wenny Rahayu}, \bibinfo{person}{Farhad Farokhi}, {and} \bibinfo{person}{Zihuai Lin}.} \bibinfo{year}{2021}\natexlab{}.
\newblock \showarticletitle{When machine learning meets privacy: A survey and outlook}.
\newblock \bibinfo{journal}{\emph{ACM Computing Surveys (CSUR)}} \bibinfo{volume}{54}, \bibinfo{number}{2} (\bibinfo{year}{2021}), \bibinfo{pages}{1--36}.
\newblock


\bibitem[\protect\citeauthoryear{Lu and Miklau}{Lu and Miklau}{2014}]%
        {lu2014exponential}
\bibfield{author}{\bibinfo{person}{Wentian Lu} {and} \bibinfo{person}{Gerome Miklau}.} \bibinfo{year}{2014}\natexlab{}.
\newblock \showarticletitle{Exponential random graph estimation under differential privacy}. In \bibinfo{booktitle}{\emph{Proceedings of the 20th ACM SIGKDD International Conference on Knowledge Discovery and Data Mining}}. ACM, \bibinfo{pages}{921--930}.
\newblock


\bibitem[\protect\citeauthoryear{Manton}{Manton}{2010}]%
        {manton2010national}
\bibfield{author}{\bibinfo{person}{Kenneth~G Manton}.} \bibinfo{year}{2010}\natexlab{}.
\newblock \showarticletitle{National Long-Term Care Survey: 1982, 1984, 1989, 1994, 1999, and 2004}.
\newblock \bibinfo{journal}{\emph{Inter-university Consortium for Political and Social Research}} (\bibinfo{year}{2010}).
\newblock


\bibitem[\protect\citeauthoryear{Marsden}{Marsden}{1993}]%
        {marsden1993reliability}
\bibfield{author}{\bibinfo{person}{Peter~V Marsden}.} \bibinfo{year}{1993}\natexlab{}.
\newblock \showarticletitle{The reliability of network density and composition measures}.
\newblock \bibinfo{journal}{\emph{Social Networks}} \bibinfo{volume}{15}, \bibinfo{number}{4} (\bibinfo{year}{1993}), \bibinfo{pages}{399--421}.
\newblock


\bibitem[\protect\citeauthoryear{McKenna and Sheldon}{McKenna and Sheldon}{2020}]%
        {mckenna2020permute}
\bibfield{author}{\bibinfo{person}{Ryan McKenna} {and} \bibinfo{person}{Daniel~R Sheldon}.} \bibinfo{year}{2020}\natexlab{}.
\newblock \showarticletitle{Permute-and-Flip: A new mechanism for differentially private selection}.
\newblock \bibinfo{journal}{\emph{Advances in Neural Information Processing Systems}}  \bibinfo{volume}{33} (\bibinfo{year}{2020}).
\newblock


\bibitem[\protect\citeauthoryear{McSherry and Talwar}{McSherry and Talwar}{2007}]%
        {mcsherry2007mechanism}
\bibfield{author}{\bibinfo{person}{Frank McSherry} {and} \bibinfo{person}{Kunal Talwar}.} \bibinfo{year}{2007}\natexlab{}.
\newblock \showarticletitle{Mechanism Design via Differential Privacy}. In \bibinfo{booktitle}{\emph{48th Annual IEEE Symposium on Foundations of Computer Science (FOCS'07)}}. \bibinfo{pages}{94--103}.
\newblock
\urldef\tempurl%
\url{https://doi.org/10.1109/FOCS.2007.66}
\showDOI{\tempurl}


\bibitem[\protect\citeauthoryear{Narayanan and Shmatikov}{Narayanan and Shmatikov}{2016}]%
        {narayanan2016break}
\bibfield{author}{\bibinfo{person}{Arvind Narayanan} {and} \bibinfo{person}{Vitaly Shmatikov}.} \bibinfo{year}{2016}\natexlab{}.
\newblock \showarticletitle{How to break anonymity of the netflix prize dataset (2006)}.
\newblock \bibinfo{journal}{\emph{arXiv preprint cs/0610105}} (\bibinfo{year}{2016}).
\newblock


\bibitem[\protect\citeauthoryear{Nissim, Raskhodnikova, and Smith}{Nissim et~al\mbox{.}}{2007}]%
        {nissim2007smooth}
\bibfield{author}{\bibinfo{person}{Kobbi Nissim}, \bibinfo{person}{Sofya Raskhodnikova}, {and} \bibinfo{person}{Adam Smith}.} \bibinfo{year}{2007}\natexlab{}.
\newblock \showarticletitle{Smooth sensitivity and sampling in private data analysis}. In \bibinfo{booktitle}{\emph{Proceedings of the thirty-ninth annual ACM symposium on Theory of computing}}. ACM, \bibinfo{pages}{75--84}.
\newblock


\bibitem[\protect\citeauthoryear{Peng, Zhou, Cao, Yu, Niu, and Jia}{Peng et~al\mbox{.}}{2018}]%
        {peng2018influence}
\bibfield{author}{\bibinfo{person}{Sancheng Peng}, \bibinfo{person}{Yongmei Zhou}, \bibinfo{person}{Lihong Cao}, \bibinfo{person}{Shui Yu}, \bibinfo{person}{Jianwei Niu}, {and} \bibinfo{person}{Weijia Jia}.} \bibinfo{year}{2018}\natexlab{}.
\newblock \showarticletitle{Influence analysis in social networks: A survey}.
\newblock \bibinfo{journal}{\emph{Journal of Network and Computer Applications}}  \bibinfo{volume}{106} (\bibinfo{year}{2018}), \bibinfo{pages}{17--32}.
\newblock


\bibitem[\protect\citeauthoryear{Quinlan}{Quinlan}{1986}]%
        {quinlan1986induction}
\bibfield{author}{\bibinfo{person}{J.~Ross Quinlan}.} \bibinfo{year}{1986}\natexlab{}.
\newblock \showarticletitle{Induction of decision trees}.
\newblock \bibinfo{journal}{\emph{Machine learning}} \bibinfo{volume}{1}, \bibinfo{number}{1} (\bibinfo{year}{1986}), \bibinfo{pages}{81--106}.
\newblock


\bibitem[\protect\citeauthoryear{Quinlan}{Quinlan}{2014}]%
        {quinlan2014c4}
\bibfield{author}{\bibinfo{person}{J~Ross Quinlan}.} \bibinfo{year}{2014}\natexlab{}.
\newblock \bibinfo{booktitle}{\emph{C4. 5: programs for machine learning}}.
\newblock \bibinfo{publisher}{Elsevier}.
\newblock


\bibitem[\protect\citeauthoryear{Series}{Series}{2015}]%
        {series2015version}
\bibfield{author}{\bibinfo{person}{Integrated Public Use~Microdata Series}.} \bibinfo{year}{2015}\natexlab{}.
\newblock \showarticletitle{Version 6.0}.
\newblock \bibinfo{journal}{\emph{Minneapolis: University of}} (\bibinfo{year}{2015}).
\newblock


\bibitem[\protect\citeauthoryear{Shokri, Stronati, Song, and Shmatikov}{Shokri et~al\mbox{.}}{2017}]%
        {shokri2017membership}
\bibfield{author}{\bibinfo{person}{Reza Shokri}, \bibinfo{person}{Marco Stronati}, \bibinfo{person}{Congzheng Song}, {and} \bibinfo{person}{Vitaly Shmatikov}.} \bibinfo{year}{2017}\natexlab{}.
\newblock \showarticletitle{Membership inference attacks against machine learning models}. In \bibinfo{booktitle}{\emph{2017 IEEE symposium on security and privacy (SP)}}. IEEE, \bibinfo{pages}{3--18}.
\newblock


\bibitem[\protect\citeauthoryear{Task and Clifton}{Task and Clifton}{2012}]%
        {task2012guide}
\bibfield{author}{\bibinfo{person}{Christine Task} {and} \bibinfo{person}{Chris Clifton}.} \bibinfo{year}{2012}\natexlab{}.
\newblock \showarticletitle{A guide to differential privacy theory in social network analysis}. In \bibinfo{booktitle}{\emph{2012 IEEE/ACM International Conference on Advances in Social Networks Analysis and Mining}}. IEEE, \bibinfo{pages}{411--417}.
\newblock


\bibitem[\protect\citeauthoryear{Zhang, Cormode, Procopiuc, Srivastava, and Xiao}{Zhang et~al\mbox{.}}{2015}]%
        {zhang2015private}
\bibfield{author}{\bibinfo{person}{Jun Zhang}, \bibinfo{person}{Graham Cormode}, \bibinfo{person}{Cecilia~M Procopiuc}, \bibinfo{person}{Divesh Srivastava}, {and} \bibinfo{person}{Xiaokui Xiao}.} \bibinfo{year}{2015}\natexlab{}.
\newblock \showarticletitle{Private release of graph statistics using ladder functions}. In \bibinfo{booktitle}{\emph{Proceedings of the 2015 ACM SIGMOD international conference on management of data}}. ACM, \bibinfo{pages}{731--745}.
\newblock


\bibitem[\protect\citeauthoryear{Zhang, Xiao, Yang, Zhang, and Winslett}{Zhang et~al\mbox{.}}{2013}]%
        {zhang2013privgene}
\bibfield{author}{\bibinfo{person}{Jun Zhang}, \bibinfo{person}{Xiaokui Xiao}, \bibinfo{person}{Yin Yang}, \bibinfo{person}{Zhenjie Zhang}, {and} \bibinfo{person}{Marianne Winslett}.} \bibinfo{year}{2013}\natexlab{}.
\newblock \showarticletitle{Privgene: differentially private model fitting using genetic algorithms}. In \bibinfo{booktitle}{\emph{Proceedings of the 2013 ACM SIGMOD International Conference on Management of Data}}. \bibinfo{pages}{665--676}.
\newblock


\bibitem[\protect\citeauthoryear{Zhao}{Zhao}{2007}]%
        {zhao2007multi}
\bibfield{author}{\bibinfo{person}{Huimin Zhao}.} \bibinfo{year}{2007}\natexlab{}.
\newblock \showarticletitle{A multi-objective genetic programming approach to developing Pareto optimal decision trees}.
\newblock \bibinfo{journal}{\emph{Decision Support Systems}} \bibinfo{volume}{43}, \bibinfo{number}{3} (\bibinfo{year}{2007}), \bibinfo{pages}{809--826}.
\newblock


\bibitem[\protect\citeauthoryear{Zhu, Li, Zhou, and Philip}{Zhu et~al\mbox{.}}{2017}]%
        {zhu2017differentially}
\bibfield{author}{\bibinfo{person}{Tianqing Zhu}, \bibinfo{person}{Gang Li}, \bibinfo{person}{Wanlei Zhou}, {and} \bibinfo{person}{S~Yu Philip}.} \bibinfo{year}{2017}\natexlab{}.
\newblock \showarticletitle{Differentially private data publishing and analysis: A survey}.
\newblock \bibinfo{journal}{\emph{IEEE Transactions on Knowledge and Data Engineering}} \bibinfo{volume}{29}, \bibinfo{number}{8} (\bibinfo{year}{2017}), \bibinfo{pages}{1619--1638}.
\newblock


\bibitem[\protect\citeauthoryear{Zitzler, Deb, and Thiele}{Zitzler et~al\mbox{.}}{2000}]%
        {zitzler2000comparison}
\bibfield{author}{\bibinfo{person}{Eckart Zitzler}, \bibinfo{person}{Kalyanmoy Deb}, {and} \bibinfo{person}{Lothar Thiele}.} \bibinfo{year}{2000}\natexlab{}.
\newblock \showarticletitle{Comparison of multiobjective evolutionary algorithms: Empirical results}.
\newblock \bibinfo{journal}{\emph{Evolutionary computation}} \bibinfo{volume}{8}, \bibinfo{number}{2} (\bibinfo{year}{2000}), \bibinfo{pages}{173--195}.
\newblock


\end{thebibliography}

\nottoggle{vldb}{
  \appendix                                     
\section{Proofs}
\label{sec:proofs}

\subsection{Proof of Theorem \ref{theo:admissible_delta_ps}}

\theoremadmissibledeltapareto*

\begin{proof} 
    According to the definition of admissibility (Definition \ref{def:admissible_function}), we first need to show that $\delta^{PS}(x, 0, r) \geq LS^{PS}(x, 0, r)$, for all $x \in \domain$ and all $r \in \range$. Thus, given a dataset $x \in \domain$ and a candidate $r \in \range$, we have:

    \begin{align*}
        LS^{PS}(x,0,r) & = \max_{z \in \domain | d(x,z) \leq 1}|PS(x,r)-PS(z,r)| \\
        &  = \max_{z \in \domain | d(x,z) \leq 1} |-|dom_{x,t,r}| + |dom_{z,t,r}| |\\
        &  = \max_{z \in \domain | d(x,z) \leq 1} max \{ |dom_{z,t,r}| - |dom_{x,t,r}|,  \\
        & \quad \quad \quad \quad \quad \quad \quad \quad  \; \; \; |dom_{x,t,r}| - |dom_{z,t,r}| \} \\
    \end{align*}

    The second equality is implied by the definition of the Pareto Score. The third equality given by the definition of the absolute value. Let's analyze the first case of the maximum function, $|dom_{z,t,r}| - |dom_{x,t,r}|$, where we maximize with respect to $z$:

    \begin{align}
        & |dom_{z,t,r}|  - |dom_{x,t,r}| \label{line:teo12_1}  \\
        & =  |dom_{z,t,r} \setminus dom_{x,t,r}| - | dom_{x,t,r} \setminus dom_{z,t,r} | \label{line:teo12_2} \\
        & \leq  |dom_{z,t,r} \setminus dom_{x,t,r}| \\
        & = |\{r' \in \range \; | \; r' \succeq^z r \; and \; r' \nsucceq^x r\}|  \\ 
        & = |\{r' \in ndom_{x,0,r} \; | \; u_i(z,r') \geq u_i(z, r), \forall i \in [m]  \}| \label{line:teo12_5} \\
        & \leq |\{r' \in ndom_{x,0,r} \; | \; u_i(z,r')+LS^{u_i}(z,0,r') \geq  \nonumber \\ 
        & \quad \quad \quad \quad \quad \quad \quad \quad \;\; u_i(z, r)-LS^{u_i}(z,0,r), \forall i \in [m]  \}| \label{line:teo12_6} \\
        & \leq |\{r' \in ndom_{x,0,r} \; | \; u_i(z,r')+\delta^{u_i}(z,0,r') \geq  \nonumber \\ 
        & \quad \quad \quad \quad \quad \quad \quad \quad \; \; u_i(z, r)-\delta^{u_i}(z,0,r), \forall i \in [m]  \}| \label{line:teo12_8} \\
        & = | ndom^+_{x,0,r} | \leq |dom^-_{x,0,r}| + |ndom^+_{x,0,r}| = \delta^{PS}(x,0,r) \label{line:teo12_9}
    \end{align}
    Line \ref{line:teo12_5} is due to the definition of $ndom$. The inequality from Line \ref{line:teo12_6} is given by the fact that $LS^u(z,0,r') \geq 0$ and $LS^u(z,0,r) \geq 0$ which implies that the set on the left-hand side is a subset of the right-hand side. line \ref{line:teo12_8} is given by the admissibility of $\delta^{u_1},\cdots,\delta^{u_m}$. The definition of $ndom^+$ results in Line \ref{line:teo12_9}. The second case of the maximum function, $|dom_{x,t,r}| - |dom_{z,t,r}|$, is symmetric to the first case. Thus, we have shown that $\delta^{PS}(x, 0, r) \geq LS^{PS}(x, 0, r)$.

    Now, We also show that $\delta^{PS}(x, t+1, r) \geq \delta^{PS}(y, t, r)$, for all $x,y$ such that $d(x,y) \leq 1$ and all $t \geq 1$. Given two neighboring datasets $x,y$, $d(x,y) \leq 1$, a distance $t$ and a candidate $r$. We have:
    \begin{align}
        & \delta^{PS}(x, t+1, r) \nonumber \\
        & = |dom^-_{x,t+1,r}| + |ndom^+_{x,t+1,r}| \nonumber \\
        & = |\{ r' \in dom_{x,r} \; | \; \exists i \in [m] \; s.t. \; u_i^{-(t+1)}(x,r') \leq  u_i^{+(t+1)}(x,r) \}| \nonumber \\
        &  + |\{ r' \in ndom_{x,r} \; | \; u_i^{+(t+1)}(x,r')  \geq  u_i^{-(t+1)}(x,r), \forall i \in [m] \} \nonumber \\
        & = |\{ r' \in (dom_{x,r}\cap dom_{y,r}) \; | \exists i \in [m] \; s.t. \nonumber \\
        & \quad \quad \quad \quad \quad \quad \quad u_i^{-(t+1)}(x,r') \leq  u_i^{+(t+1)}(x,r) \}| \nonumber \\
        & + |\{ r' \in (dom_{x,r}\cap ndom_{y,r}) \; | \; \exists i \in [m] \; s.t. \nonumber\\
        & \quad \quad \quad \quad \quad \quad \quad u_i^{-(t+1)}(x,r') \leq  u_i^{+(t+1)}(x,r) \}| \nonumber \\     
        & + |\{ r' \in (ndom_{x,r} \cap dom_{y,r})  \; | \nonumber \\
        & \quad \quad \quad \quad \quad \quad \quad u_i^{+(t+1)}(x,r')  \geq  u_i^{-(t+1)}(x,r), \forall i \in [m] \}| \nonumber \\
        & + |\{ r' \in (ndom_{x,r} \cap ndom_{y,r})  \; | \nonumber \\ 
        & \quad \quad \quad \quad \quad \quad \quad  u_i^{+(t+1)}(x,r')  \geq  u_i^{-(t+1)}(x,r), \forall i \in [m] \}| \label{line:four_terms}
    \end{align}

    The last equality is true since $dom_{x,r} = (dom_{x,r}\cap dom_{y,r}) \cup (dom_{x,r}\cap ndom_{y,r})$ as $dom_{y,r} \cup ndom_{y,r} = \range$  which also implies that $ndom_{y,r} = (ndom_{x,r}\cap dom_{y,r}) \cup (ndom_{x,r}\cap dom_{y,r})$. Also, note that $(dom_{x,r}\cap dom_{y,r}) \cap (dom_{x,r}\cap ndom_{y,r}) = \emptyset$ as $dom_{y,r} \cap ndom_{y,r} = \emptyset$ and $(ndom_{x,r}\cap dom_{y,r}) \cap (ndom_{x,r}\cap ndom_{y,r}) = \emptyset$ for the same reason.

    Now, we analyze each term of the sum in Line \ref{line:four_terms}. The first term is given by:
    \begin{align*}
        & |\{ r' \in (dom_{x,r} \cap dom_{y,r}) \; | \; \exists i \in [m] \; s.t. \; \\
        & \quad \quad \quad u_i(x,r') - \sum_{j=0}^{t+1} \delta^{u_i}(x,j,r') \leq  u_i(x,r) + \sum_{j=0}^{t+1} \delta^{u_i}(x,j,r) \}| \\
        & \geq |\{ r' \in (dom_{x,r} \cap dom_{y,r}) \; | \; \exists i \in [m] \; s.t. \; \\
        & \quad \quad \quad u_i(y,r') - \sum_{j=1}^{t+1} \delta^{u_i}(x,j,r') \leq  u_i(y,r) + \sum_{j=1}^{t+1} \delta^{u_i}(x,j,r) \}| \\
        & \geq |\{ r' \in (dom_{x,r} \cap dom_{y,r}) \; | \; \exists i \in [m] \; s.t. \; \\
        & \quad \quad \quad u_i(y,r') - \sum_{j=0}^{t} \delta^{u_i}(x,j,r') \leq  u_i(y,r) + \sum_{j=0}^{t} \delta^{u_i}(x,j,r) \}| \\
        & = |\{ r' \in (dom_{x,r} \cap dom_{y,r}) \; | \; \exists i \in [m] \; s.t. \; u_i^{+t}(y,r') \leq  u_i^{-t}(y,r) \}|
    \end{align*}
    The first inequality is true as $u_i(x, r') \leq u_i(y, r') + \delta^{u_i}(y,0,r')$ and $u_i(x, r) \geq u_i(y, r) + \delta^{u_i}(y,0,r)$ which implies that the set on the left-hand side is a superset of the right-hand side. The second inequality is due to the admissibility of $\delta^{u_i}$.
    
    Now, we analyze the third term of the sum. Let $r' \in dom_{y,r}$ be a candidate, then $u_i(y,r') \geq u_i(y,r), \forall i \in [m]$. We have that:
    

    \begin{align}
        & u_i(y,r') \geq u_i(y,r), \forall i \in [m] \label{line:condition1} \\
        & \Rightarrow u_i(x,r') + \delta^{u_i}(x,0,r') \geq u_i(x,r) - \delta^{u_i}(x,0,r), \forall i \in [m] \\
        & \Rightarrow u_i(x,r') + \sum_{j=0}^t \delta^{u_i}(x,t,r') \geq u_i(x,r) - \sum_{j=0}^t \delta^{u_i}(x,t,r), \forall i \in [m] \label{line:condition3} 
    \end{align}

    The first implication is true as $u_i(y, r') \leq u_i(x, r') + \delta^{u_i}(x,0,r')$ and $u_i(y, r) \geq u_i(x, r) + \delta^{u_i}(x,0,r)$. The second implication is due to the fact that $\delta^{u_i}(x,t,r') \geq LS^{u_i}(x,t,r') \geq 0$ and $\delta_{u_i}(x,t,r) \geq LS^{u_i}(x,t,r) \geq 0$, $\forall t \geq 0$. Thus the condition on Line \ref{line:condition1} implies the condition on Line \ref{line:condition3}. Therefore by the absorption law:

    \begin{align*}
    & \{ r' \in (ndom_{x,r} \cap dom_{y,r})  \; | \; u_i^{+(t+1)}(x,r')  \geq  u_i^{-(t+1)}(x,r), \forall i \in [m] \} \\
    & = \{ r' \in (ndom_{x,r} \cap dom_{y,r}) \} \\
    & \geq \{  r' \in (ndom_{x,r} \cap dom_{y,r}) \; | \; \exists i \in [m] \; s.t. \; u_i^{+t}(y,r') \leq  u_i^{-t}(y,r) \}
    \end{align*}
    
    The inequality is true as the set on the left-hand side is a superset of the right-hand side. The second and the forth terms of the sum are analyzed in a similar way. Following from Line \ref{line:four_terms}, we have that:

    \begin{align*}
        & \delta^{PS}(x, t+1, r) \\
        & \geq |\{ r' \in (dom_{x,r} \cap dom_{y,r}) \; | \; \exists i \in [m] \; s.t. \; u_i^{+t}(y,r') \leq  u_i^{-t}(y,r) \}| \\
        & + |\{ r' \in (dom_{x,r}\cap ndom_{y,r}) \; | \; u_i^{+t}(y,r')  \geq  u_i^{-t}(y,r), \forall i \in [m] \}| \\   
        & + |\{  r' \in (ndom_{x,r} \cap dom_{y,r}) \; | \; \exists i \in [m] \; s.t. \; u_i^{+t}(y,r') \leq  u_i^{-t}(y,r) \} | \\
        & + |\{ r' \in (ndom_{x,r} \cap ndom_{y,r})  \; | \; u_i^{+t}(y,r')  \geq  u_i^{-t}(y,r), \forall i \in [m] \}| \\
        & = |dom^-_{y,t,r}| + |ndom^+_{y,t,r}| = \delta^{PS}(y,t,r)
    \end{align*}

    since $dom_{x,r} \cup ndom_{x,r} = \range$ and $dom_{y,r} \cup ndom_{y,r} = \range$. Therefore, $\delta^{PS}(x, t+1, r) \geq \delta^{PS}(y, t, r)$, for all $x,y$ such that $d(x,y) \leq 1$ and all $t \geq 1$. Thus, $\delta^{PS}(x,t,r)$ is admissible.
\end{proof}

\subsection{Proof of Lemma \ref{lemma:admissible_sensitivity_tpr_tnr}}


    \begin{equation*}
        \delta^{TPR}(x,t,\tree) =
        \begin{cases}
            \frac{P'(x, t)-TP'(x, t, \tree)}{P'(x, t)(P'(x, t)-1)} & \text{if } t \leq P(x) -TP(x,\tree) \\
            \frac{TP'(x, t, \tree)}{P'(x, t)(P'(x, t)-1)} & \text{otherwise}
        \end{cases}, 
    \end{equation*}
    

    \begin{equation*}
        P'(x,t,\tree) =
        \begin{cases}
            P(x) - t & \text{if } t \leq P(x) - 2 \\
            2 & \text{otherwise}
        \end{cases}, 
    \end{equation*}

    \begin{equation*}
        TP'(x,t,\tree) =
        \begin{cases}
            TP(x,\tree) & \text{if } t \leq P(x) -TP(x,\tree) \\
            P(x) - t & \text{if } t \leq P(x) - 2\\
            2 & \text{otherwise}
        \end{cases}, 
    \end{equation*}


\theoremadmissibledeltatprtnr*

\begin{proof}
    In this proof, we show that $\delta^{TPR}(x,t,\tree)$ is admissible and the proof $\delta^{TNR}(x,t,\tree)$ is symmetric. 

    For this proof, suppose that $TP(x,\tree) \geq 2$ and $P(x) \geq 2$. First, we show that $\delta^{TPR}(x,0,\tree) \geq LS^{TPR}(x,0,\tree)$, for all $x \in \domain$ and all $\tree \in \range$. Given a dataset $x \in \domain$, a neighboring dataset $x'$ of $x$ and a candidate $\tree \in \range$, we have that:



    \begin{align*}
        LS^{TPR}(x,0,\tree) & = \max_{x' \in \domain | d(x,x') \leq 1} \vast| \frac{TP(x,\tree)}{P(x)} - \frac{TP(x',\tree)}{P(x')} \vast| \\
        & = \max \vast\{ \underbrace{ \left| \frac{TP(x,\tree)}{P(x)} - \frac{TP(x,\tree)+1}{P(x)+1} \right| }_{(1)} , \\
        & , \underbrace{ \left| \frac{TP(x,\tree)}{P(x)} - \frac{TP(x,\tree)}{P(x)+1} \right|}_{(2)}, \underbrace{\left| \frac{TP(x,\tree)}  {P(x)} - \frac{TP(x,\tree)-1}{P(x)-1} \right|}_{(3)},    \\
        & , \underbrace{\left| \frac{TP(x,\tree)}{P(x)} - \frac{TP(x,\tree)}{P(x)-1} \right|}_{(4)}, \underbrace{\left| \frac{TP(x,\tree)}{P(x)} - \frac{TP(x,\tree)}{P(x)} \right|}_{(5)} \vast\}
    \end{align*}

    The last equality is given by the analysis of $x'$. The database $x'$ can be obtained from $x$ in one of the following ways:
    
    \begin{enumerate}[label=\roman*]
        \item a true positive is added to $x$ to obtain $x'$, which implies that $TP(x', \tree) = TP(x', \tree) + 1$ and $P(x') = P(x) + 1$ (1); 
        \item a false negative is added to $x$ to obtain $x'$, which implies that $TP(x', \tree) = TP(x', \tree)$ and $P(x') = P(x) + 1$ (2); 
        \item a true positive is removed from $x$ to obtain $x'$, which implies that $TP(x', \tree) = TP(x', \tree) - 1$ and $P(x') = P(x) - 1$ (3); 
        \item a false negative is removed from $x$ to obtain $x'$, which implies that $TP(x', \tree) = TP(x', \tree)$ and $P(x') = P(x) - 1$ (4); 
        \item a false positive is added to $x$ to obtain $x'$, which implies that $TP(x', \tree) = TP(x', \tree)$ and $P(x') = P(x)$ (5).
        \item a false positive is removed from $x$ to obtain $x'$, which implies that $TP(x', \tree) = TP(x', \tree)$ and $P(x') = P(x)$ (5).
        \item a true negative is added to $x$ to obtain $x'$, which implies that $TP(x', \tree) = TP(x', \tree)$ and $P(x') = P(x)$ (5).
        \item a true negative is removed from $x$ to obtain $x'$, which implies that $TP(x', \tree) = TP(x', \tree)$ and $P(x') = P(x)$ (5).
    \end{enumerate}

    To solve the maximization, we state the following facts: a) (3) is larger than (1) as (3) can be simplified to $\frac{P(x)-TP(x, \tree)}{P(x)(P(x)-1)}$ and (1) can be simplified to $\frac{P(x)-TP(x, \tree)}{P(x)(P(x)+1)}$; b) (4) is larger than (2) since (4) can be reduced to $\frac{TP(x, \tree)}{P(x)(P(x)+1)}$ and (2) to $\frac{TP(x, \tree)}{P(x)(P(x)-1)}$; and c) (5) is equal to $0$.

    Thus,

    \begin{align*}
        LS^{TPR}(x,0,\tree) & = \max \vast\{ \frac{P(x)-TP(x, \tree)}{P(x)(P(x)-1)}, \frac{TP(x, \tree)}{P(x)(P(x)-1)} \vast\} \\
        & = delta^{TPR}(x,0,\tree)
    \end{align*}





\end{proof}

\subsection{Proof of Lemma \ref{lemma:egodensity_sensitivity}}

\theoremadmissibledeltaegodensity*

\begin{proof}

    Let $A(G, v) = 2.|\{(u,w) \in E(G): u,w \in N^G(v)\}|$. First, we show that $\delta^{egodensity}(x, 0, r) \geq LS^{egodensity}(x, 0, r)$, for all $x \in \domain$ and all $r \in \range$. Here, we consider that $|N^G(v)| > 2$ as the case where $|N^G(v)| \leq 2$ is straightforward. Given a dataset $x \in \domain$ and a candidate $r \in \range$, we have:
    
    \begin{align*}
        LS^{ed}(G,0,v) & = \max_{G' \in \domain | d(G,G') \leq 1}|egodensity(G,v) \\
        & \quad \quad \quad \quad \quad \quad \quad -egodensity(G',v)| \\
        & = \max_{G' \in \domain | d(G,G') \leq 1} \vast| \frac{2.A(G, v)}{|N^G(v)| . (|N^G(v)|-1)}  \\
        & \quad \quad \quad \quad \quad \quad \; - \frac{2.A(G', v)}{|N^{G'}(v)| . (|N^{G'}(v)|-1)} \vast| \\
        & = \max_{G' \in \domain | d(G,G') \leq 1} \max \vast\{ \frac{2.A(G, v)}{|N^G(v)| . (|N^G(v)|-1)}  \\ 
        & \quad \quad \quad \quad \quad \quad \quad \quad \quad - \frac{2.A(G', v)}{|N^{G'}(v)| . (|N^{G'}(v)|-1)}, \\
        & \quad \quad \quad \quad \quad \quad \quad \quad \quad \quad  \frac{2.A(G', v)}{|N^{G'}(v)| . (|N^{G'}(v)|-1)}\\ 
        & \quad \quad \quad \quad \quad \quad \quad \quad \quad  - \frac{2.A(G, v)}{|N^G(v)| . (|N^G(v)|-1)} \vast\} \\
        & = \max \vast\{ \max_{G' \in \domain | d(G,G') \leq 1} \vast( \frac{2.A(G, v)}{|N^G(v)| . (|N^G(v)|-1)}  \\ 
        & \quad \quad \quad \quad \quad \quad \quad \quad \quad \quad - \frac{2.A(G', v)}{|N^{G'}(v)| . (|N^{G'}(v)|-1)} \vast), \\
        & \quad \quad \quad \quad  \max_{G' \in \domain | d(G,G') \leq 1} \vast(   \frac{2.A(G', v)}{|N^{G'}(v)| . (|N^{G'}(v)|-1)}\\ 
        & \quad \quad \quad \quad \quad \quad \quad \quad \quad \quad  - \frac{2.A(G, v)}{|N^G(v)| . (|N^G(v)|-1)} \vast) \vast\}
    \end{align*}

Thus to to maximize the outer maximum function of the last line, we maximize each of the inner maximum functions. The graph $G'$ is a neighboring graph of $G$ that differs by exactly one edge. Let $e=(u,w)$ be the edge that is added or removed from $G$ to obtain $G'$. Let's break down the proof in two cases: (i) Suppose that edge $e$ is added to $G$ to obtain $G'$, i.e., $E(G') = E(G) \cup \{e\}$; (ii) Suppose that edge $e$ is removed from $G$ to obtain $G'$, i.e., $E(G') = E(G) \setminus \{e\}$.

\begin{enumerate}[label=(\roman*)]    

    \item First, suppose that $e$ is added to $G$ to obtain $G'$. Let's analyze the three cases for the edge $e$: 1) $u \neq v$ or $w \neq v$ are not neighbors of $v$ in $G$, implying that $A(G, v) = A(G', v)$ and $N^{G}(v) = N^{G'}(v)$; 2) $u \neq v$ and $w \neq v$ are neighbors of $v$ in $G$, which means that $N^{G'}(v) = N^{G}(v)$ and $A(G', v) = A(G, v) + 1$; 3) $u \neq v$ and $w = v$ ($w$ is not a neighbor of $v$ in $G$) which implies that $N^{G'}(v) = N^{G}(v) +1$ and $A(G, v) \leq A(G', v) \leq A(G, v) + |N^{G}(v)|$ as $w$ can be connected to at most $|N^{G}(v)|$ nodes in $G'$ that are also connected to $v$; 4) The case where $u \neq v$ and $w = v$ is symmetric to the previous one. Therefore, in general, we have that $A(G, v) \leq A(G', v) \leq A(G, v) + |N^{G}(v)|$ and $|N^{G}(v)| \leq |N^{G'}(v)| \leq |N^{G}(v)| + 1$. Therefore, we calculate and upper bound for both inner maximum functions:
    
    \begin{align*}
        & \frac{2.A(G, v)}{|N^G(v)| . (|N^G(v)|-1)} - \frac{2.A(G', v)}{|N^{G'}(v)| . (|N^{G'}(v)|-1)} \\
        & \leq \frac{2.A(G, v)}{|N^G(v)| . (|N^G(v)|-1)} - \frac{2.A(G, v)}{(|N^G(v)|+1) . (|N^G(v)|)} \\
        & = \frac{2}{|N^G(v)| }. \left( \frac{A(G, v)}{|N^G(v)|-1)} - \frac{A(G, v)}{|N^G(v)|+1)}  \right) \\
        & = \frac{2}{|N^G(v)| }. \left( \frac{2.A(G, v)}{(|N^G(v)|-1).(|N^G(v)|+1)}  \right) \\
        & \leq \frac{2}{|N^G(v)| }. \left( \frac{(|N^G(v)|-1).|N^G(v)|}{(|N^G(v)|-1).(|N^G(v)|+1)}  \right) \\
        & = \frac{2}{|N^G(v)| }. \frac{|N^G(v)|}{|N^G(v)|+1} \\
        & = \frac{2}{|N^G(v)| +1} \leq \delta^{egodensity}(G,0,v)
    \end{align*}
    since $A(G,v) \leq  \binom{|N^G(v)|}{2} $. Then:
    \begin{align*}
        & \frac{2.A(G', v)}{|N^{G'}(v)| . (|N^{G'}(v)|-1)} - \frac{2.A(G, v)}{|N^G(v)| . (|N^G(v)|-1)} \\
        & \leq \frac{2.(A(G, v) + |N^{G}(v)|)}{|N^{G}(v)| . (|N^{G}(v)|-1)} - \frac{2.A(G, v)}{|N^G(v)| . (|N^G(v)|-1)} \\
        & = \frac{2.}{|N^{G}(v)|-1} \leq \delta^{egodensity}(G,0,v)
    \end{align*}

    \item Now, suppose that edge $e$ is removed from $G$ to obtain $G'$. Let's analyze the four cases for the edge $e$: 1) $u \neq v$ or $w \neq v$ are not neighbors of $v$ in $G$, implying that $A(G, v) = A(G', v)$ and $N^{G}(v) = N^{G'}(v)$; 2) $u \neq v$ and $w \neq v$ are neighbors of $v$ in $G$, which means that $N^{G'}(v) = N^{G}(v)$ and $A(G', v) = A(G, v) - 1$; 3) $u \neq v$ and $w = v$ ($w$ is a neighbor of $v$ in $G$) which implies that $N^{G'}(v) = N^{G}(v) -1$ and $A(G, v) - (|N^{G}(v)|-1) \leq A(G', v) \leq A(G, v)$ as $w$ can be connected to at most $|N^{G}(v)|-1$ nodes in $G$ that are also connected to $v$; 4) The case where $u \neq v$ and $w = v$ is symmetric to the previous one. Thus, in general, we have that $A(G, v) - (|N^{G}(v)|-1) \leq A(G', v) \leq A(G, v)$ and $|N^{G'}(v)| -1 \leq |N^{G}(v)| \leq |N^{G'}(v)|$. Therefore, we calculate and upper bound for both inner maximum functions:
    \begin{align*}
        & \frac{2.A(G, v)}{|N^G(v)| . (|N^G(v)|-1)} - \frac{2.A(G', v)}{|N^{G'}(v)| . (|N^{G'}(v)|-1)} \\
        & \leq \frac{2.A(G, v)}{|N^G(v)| . (|N^G(v)|-1)} - \frac{2.(A(G, v)- (|N^{G}(v)|-1))}{|N^G(v)| . (|N^G(v)|-1)} \\
        & = \frac{2}{|N^G(v)| } \leq \delta^{egodensity}(G,0,v)
    \end{align*}
    and
    \begin{align*}
        & \frac{2.A(G', v)}{|N^{G'}(v)| . (|N^{G'}(v)|-1)} - \frac{2.A(G, v)}{|N^G(v)| . (|N^G(v)|-1)} \\
        & \leq \frac{2.A(G, v) |}{(|N^{G}(v)|-1) . (|N^{G}(v)|-2)} - \frac{2.A(G, v)}{|N^G(v)| . (|N^G(v)|-1)} \\
        & \leq \frac{2}{|N^{G}(v)|-1} \left (\frac{A(G, v) |}{(|N^{G}(v)|-2)} - \frac{A(G, v)}{|N^G(v)| } \right) \\
        & = \frac{2}{|N^{G}(v)|-1} \left (\frac{2.A(G, v) |}{(|N^G(v)|.|N^{G}(v)|-2)} \right) \\
        & \leq \frac{2}{|N^{G}(v)|-1} \left (\frac{(|N^G(v)|-1).|N^G(v)|}{(|N^G(v)|.|N^{G}(v)|-2)} \right) \\
        & = \frac{2}{|N^{G}(v)|-1} . \frac{|N^G(v)|-1}{|N^{G}(v)|-2} \\
        & = \frac{2.}{|N^{G}(v)|-2} = \delta^{egodensity}(G,0,v)
    \end{align*}

    \end{enumerate}

    Now, we show that $\delta^{egodensity}(G, t+1, v) \geq \delta^{egodensity}(G, t, v)$, for all $G,G'$ such that $d(G,G') \leq 1$ and all $t \geq 1$. Here, we assume that $|N^G(v)|-t > 2$ as the case where $|N^G(v)|-t \leq 2$ is straightforward. Given two neighboring graphs $G,G'$, $d(G,G') \leq 1$, a distance $t < |N^G(v)|-2$ and a candidate vertex $v$. We have:

    \begin{align*}
        & \delta^{egodensity}(G, t+1, v)  = \frac{2}{degree(G,v)-t-3} \\        
        & \geq \frac{2}{degree(G',v)-t-2} = \delta^{egodensity}(G', t, v)
    \end{align*}
    This is true since $degree(G,v) \leq degree(G',v)+1$ as $G$ and $G'$ differ by at most one edge. 

\end{proof}

}

\end{document}